\documentclass[11pt]{article}
\usepackage[margin=0.75in]{geometry}
\usepackage{times}

\usepackage{hyperref}
\usepackage[hyphenbreaks]{breakurl}
\usepackage{amsthm}
\usepackage{amsmath,amsfonts}

\usepackage{amssymb}
\usepackage{url}
\usepackage{thmtools,thm-restate}
\usepackage{makecell}
\newtheorem{example}{Example}
\newtheorem{definition}{Definition}
\newtheorem{theorem}{Theorem}[section]
\newtheorem{lemma}[theorem]{Lemma}

\usepackage{mathtools}

\usepackage[ruled,linesnumbered,vlined]{algorithm2e}
\SetKw{KwOr}{or}
\SetKw{KwAnd}{and}

\usepackage{bm}

\makeatletter
\newsavebox{\@brx}
\newcommand{\llangle}[1][]{\savebox{\@brx}{\(\m@th{#1\langle}\)}%
	\mathopen{\copy\@brx\kern-0.5\wd\@brx\usebox{\@brx}}}
\newcommand{\rrangle}[1][]{\savebox{\@brx}{\(\m@th{#1\rangle}\)}%
	\mathclose{\copy\@brx\kern-0.5\wd\@brx\usebox{\@brx}}}
\makeatother

\usepackage{listings}
\usepackage{color}
\definecolor{dkgreen}{rgb}{0,0.6,0}
\definecolor{gray}{rgb}{0.5,0.5,0.5}
\definecolor{mauve}{rgb}{0.58,0,0.82}
\newcommand{\TOP}{\textsf{TOP}}

\newcommand{\E}{\mathbb{E}}

\newcommand{\prefix}[1]{\texttt{prefix\_{#1}}}

\usepackage{caption}
\usepackage{subcaption}
\usepackage{multirow}

\SetKwFor{ParFor}{for}{do in parallel}{end}

\usepackage{tikz}

\newcommand{\rank}[1]{\hat{#1}}
\newcommand{\bmrank}[1]{\rank{\bm{#1}}}

\begin{document}

\title{Secure Query Processing with Linear Complexity}
\author{
Qiyao Luo \thanks{Hong Kong University of Science and Technology. {\tt qluoak@cse.ust.hk}.}
\and
Yilei Wang \thanks{Alibaba Cloud. {\tt fengmi.wyl@alibaba-inc.com}}
\and
Wei Dong \thanks{Carnegie Mellon University. {\tt wdong2@andrew.cmu.edu}.}
\and
Ke Yi \thanks{Hong Kong University of Science and Technology. {\tt yike@cse.ust.hk}.}
}
\date{}

\maketitle

\begin{abstract}
We present LINQ, the first join protocol with linear complexity (in both running time and communication) under the secure multi-party computation model (MPC). It can also be extended to support all free-connex queries, a large class of select-join-aggregate queries, still with linear complexity.  This matches the plaintext result for the query processing problem, as free-connex queries are the largest class of queries known to be solvable in linear time in plaintext. We have then built a query processing system based on LINQ, and the experimental results show that LINQ significantly outperforms the state of the art.  For example, it can finish a query on three relations with an output size of 1 million tuples in around 100s in the LAN setting, while existing protocols that support the query cannot finish in an hour. Thus LINQ brings MPC query processing closer to practicality.
\end{abstract}

\section{Introduction}
With significant improvements in both computing infrastructure and protocol designs, \textit{secure multi-party computation (MPC)} is now transforming from a theoretical concept to a practical technique for breaking the barrier between information silos.  It holds the appealing guarantee that a query can be evaluated on the private data of multiple parties without revealing anything beyond the query result. This meets the needs of enterprises who want to perform collaborative analysis without sharing their private data.
\begin{example}
\label{intro_exa}
Consider a scenario where an insurance company, a hospital, and a bank would like to collaboratively analyze users' behaviour based on their insurance plans, medical history, and deposits. These datasets are usually stored in relational databases.  As an example, assume the following simplified schema: Insurance Plans {\tt I(id, insurance\_pay, start\_year, end\_year)}, Medical Records {\tt M(id, disease, treatment\_cost, year)}, and Deposits {\tt D(id, deposit)}. Then the  select-join-aggregate query below conducts a typical collaborative analysis, which  finds the average amount of medical payment of patients whose deposits fall inside a given range $[10000, 30000]$,  classified by diseases:
\begin{center}
\begin{tabular}{c}
\begin{lstlisting}
SELECT disease, AVG(M.treatment_cost - I.insurance_pay)
  FROM I, M, D
 WHERE I.id = M.id AND M.id = D.id 
   AND D.deposit >= 10000 AND D.deposit <= 30000
GROUP BY disease
\end{lstlisting}
\end{tabular}
\end{center}

Since this query is a free-connex query, it can be executed in linear time in plaintext, and all modern database systems can run such queries very efficiently.  However, they are of no help when the three relations are owned by different parties who do not trust each other.  This scenario thus calls for an efficient MPC protocol that can compute the query while protecting the privacy of the three data owners.
\qed
\end{example}

The major technical challenge in executing the query above is the join operator.
A brute-force method for computing a join of $k$ relations, each having $n$ tuples, is known as \textit{nested-loop join}.  It has a cost of $O(n^k)$, which is actually already worst-case optimal as the output size of the join, denoted as $m$, can be as large $n^k$ in the worst case.  However, worst-case optimality is not very meaningful for this problem, as $m \ll n^k$ on typical inputs.  Thus, the database literature is more interested in an $O(n+m)$ running time.  Such a running time is often called \textit{linear}, which is clearly the best one can hope for.  One landmark result in query processing is that linear time is achievable if and only if the query is \textit{free-connex} (defined formally in Section \ref{sec:freeconnex}).  The upper bound uses hash-joins and the Yannakakis algorithm \cite{yannakakis,ajar,faqpaper}, and the lower bound is based on some standard complexity assumptions \cite{10.1007/978-3-540-74915-8_18}.

Generic MPC protocols \cite{mpcbook} require the computation to be expressed as a circuit.  While nested-loop join can be trivially circuitized, its $O(n^k)$ size makes it unaffordable.  There is also a circuit version of the Yannakakis algorithm, but its size is $O((n+m) \log^4(n+m))$ \cite{wang2022query}, which is still impractical.  Therefore, recent efforts have focused on designing custom protocols, often taking advantage of a particular MPC variant, such as the two-party model (2PC) \cite{secyan} or the three-party honest majority model (3PC) \cite{secretsharedjoins,scape,secrecy}.
However, linear complexity (in either running time or communication) has yet to be achieved (see Table \ref{tab:comp} for a summary of prior results), leaving a gap from plaintext query processing. 

\subsection{Our Contributions}
We close the gap in this paper by presenting LINQ, the first \underline{LIN}ear complexity \underline{Q}uery processing protocol under the 3PC model with linear complexity in both computation and communication.  More importantly, LINQ supports all free-connex queries with linear complexity, thus matching the plaintext result on the largest class of queries.  The intriguing problem of whether linear complexity is achievable under 2PC is left open.  Below we describe some technical highlights behind our result.

We start by re-visiting the two-way join problem.
In plaintext, the hash-join algorithm computes a join in $O(n+m)$ time.  However, it is difficult to implement in MPC; intuitively, this is because hash join requires a hash table that associates multiple values to a join key (e.g., in Example \ref{intro_exa}, multiple tuples in relation {\tt M} are associated with the same join key {\tt id}), resulting in uneven, input-dependent bucket sizes.  As a result, all existing MPC join protocols \cite{secretsharedjoins,scape,senate} are based on the sort-merge-join algorithm, which inevitably incurs an $\Omega(n\log n)$ cost.   To overcome this difficulty, our idea is to combine the two join algorithms into a ``hash-sort-merge-join''. We start with the observation that sort-merge-join does not really need a true sort.  Instead, any \textit{consistent sort} suffices, where the two relations are first sorted by the same, but arbitrary, ordering of the domain of the join key.  Next, we show that a consistent sort can be done by each party in $O(n)$ time by hashing, as opposed to the usual sorting complexity $O(n\log n)$.  The results of the consist-sort are recorded in an additional column in each relation called \textit{ranks}.  Then all the columns are secret-shared among the 3 parties,
who can then perform the merge-join step with $O(n+m)$ cost using some standard 3PC primitives. 

Extending our join protocol to a free-connex query, which involves multiple join operators (plus select and aggregate), is the next technical challenge. As seen from Table \ref{tab:comp}, not many prior protocols can support multi-way joins.  Those that do are either based on the nested-loop join or incur a high polylogarithmic overhead.  The only exception is \cite{secyan}. However, it actually dodges the MPC join problem, as we explain in more detail in Section \ref{sec:related}.  We use a simple 3-way join $R_1 \Join R_2 \Join R_3$ to illustrate the technical difficulty here.  Suppose each of the 3 relations is owned by a party. 
To first compute $R_1\Join R_2$, the first two parties perform consistent sorting of $R_1$ and $R_2$ to obtain their ranks.  Then our protocol computes $R_1\Join R_2$.  However, we cannot apply the same process to compute $(R_1 \Join R_2) \Join R_3$ because the intermediate join results $R_1 \Join R_2$ are only available in secret-shared form.  In particular, it is not possible to do consistent-sort on secret-shared data with linear cost.  To overcome this difficulty, we exploit some important relationships between the ranks in the original $R_1,R_2$ and the ranks in $R_1\Join R_2$, and manage to transform the former into the latter, still with linear complexity.

Finally, in order to make LINQ practical, we need to pay attention to the hidden constant in the big-Oh.  For a given free-connex query, we first enumerate all the $O(n+m)$-complexity query plans.  Then We have designed a cost model to estimate the cost of each query plan, and choose the best one for execution.  We have built a system prototype by integrating these components.
Our experiments show that our system outperforms the state of the art systems significantly, and can finish the join of three relations that outputs 1 million tuples in around 100s in the LAN setting, thus bringing MPC query processing one step closer to practicality.

\subsection{Related Work}
\label{sec:related}

\begin{table*}
\renewcommand\arraystretch{1.2}
	\centering
	\begin{tabular}{|c|c|c|c|c|c|}
		\hline
		\multirow{2}{*}{Methods} & \#Parties / & \multirow{2}{*}{Two-way join} & \multicolumn{2}{c|}{Multi-way join on $k$ relations} \\
            \cline{4-5}
             & \#Corrupted parties & & Free-connex join & PK-FK join \\
		\hline \hline
		SMCQL \cite{smcql} & 2 / 1  & $O(n^2)$ & $O(n^k)$ & $O(n^k)$ \\
		Secrecy \cite{secrecy} & 3 / 1 &  $O(n^2)$ & $O(n^k)$ & $O(n^k)$ \\
		\hline
		SSJ \cite{secretsharedjoins} & 3 / 1 &  $O(n\log n+m\log m)$ & Unsupported & $O(n \log n)$\\
		Scape \cite{scape} & 3 / 1 & $O(n \log^2 n + m)$ & Unsupported & $O(n \log^2 n)$\\
        OptScape & 3 / 1 & $O(n \log n + m)$ & Unsupported & $O(n \log n)$ \\ \hline
        SECYAN$^*$ \cite{secyan} & 2 / 1 & $O(n\log n+m)$ & $O(n\log n+m)$ & $O(n \log n)$\\
        QCircuit \cite{wang2022query} & Any & $O((n+m)\log^4(n+m))$ & $O((n+m)\log^4(n+m))$ & $O(n \log^4 n)$\\\hline
		LINQ (Ours) & 3 / 1 & $O(n + m)$ & $O(n + m)$ & $O(n)$\\\hline
	\end{tabular}
	\caption{Comparison between LINQ and previous works, where $n$ and $m$ are the input and output size of the query, respectively.  $^*$SECYAN returns the query results to one party in plaintext.}
	\label{tab:comp}
\end{table*}

Table~\ref{tab:comp} summarizes the existing MPC join algorithms.  The first two algorithms are based on the nested-loop join, which has an $O(n^2)$ complexity.  This balloons to $O(n^k)$ for a query on $k$ relations.  Note that the output size $m$ does not show up as it is dominated by $O(n^k)$.  

The next three algorithms are based on the sort-merge-join algorithm.  The two almost concurrent works of SSJ \cite{secretsharedjoins} and Scape \cite{scape} have achieved similar, but incomparable, results.  However, we observe that Scape used the $O(n\log^2 n)$ bitonic sort instead of the optimal $O(n \log n)$ sorting algorithm \cite{hamada2013sort, efficient3pcsorting}. We have re-implemented Scape using the $O(n\log n)$ sorting protocol, and denote this optimized version as OptScape. OptScape is then better than both Scape and SSJ.   We note that these three protocols also work in the \textit{three-server} model, a variant of 3PC that takes input data in secret-shared form. This setting is stronger than 3PC, as detailed in Section~\ref{sec:3server}. However, even assuming the two input relations are already ordered by their join keys, they still require $O(n \log n + m)$ complexity.  Furthermore, these protocols do not support queries with multiple join operators; the simple composition of their two-way join protocols would reveal the intermediate join sizes instead of just the final output size, breaching the security guarantee.  Nevertheless, in the special case wheref all joins are PK-FK joins, the intermediate join sizes are bounded by the input size, so can be padded to $n$.  Thus, for multi-way PK-FK joins, this composition is secure, and the complexity becomes $O(n \log n)$.   Note that the complexity of LINQ becomes $O(n)$ on PK-FK joins as $m\le n$. 

SECYAN \cite{secyan} computes any free-connex query in $O(n\log n+m)$ time in the 2PC model, but one of the two parties will get the query results in plaintext.  This is problematic for multi-way join queries without aggregation in at least two aspects.  The first is technical.  Note that if there are no dangling tuples, i.e., tuples that do not join with any tuple in the other relations, then the multi-way join results would consist of all the input data.  Thus, it is okay to reveal all input data upfront, trivializing the problem.  So the only care that must be taken is the dangling tuples should be removed (more precisely, should be changed to dummies), and the problem essentially boils down to a series of semi-joins. This is also why they call their protocol a secure version of the Yannakakis algorithm \cite{yannakakis}, in which the central step consists of two passes of semi-joins.  As such, SECYAN does not include an MPC join algorithm per se. The second reason is the applicability to real-world scenarios.  A multi-way join query without aggregation would return too many results.  If they are all revealed, this would still lead to privacy breaches.  To fix this issue, we require the query output be given in secret-shared form.  In real applications, these secret-shared results can then be fed into a downstream application, e.g., an MPC machine learning algorithm to train a model.  We note that all the other protocols listed in Table~\ref{tab:comp}, including ours, return the query results in secret-shared form.

QCircuit \cite{wang2022query} presents a circuit that evaluates any free-connex query with size $O((n+m)\log^4(n+m))$, which can be executed with some generic MPC protocol \cite{yao1982protocols, goldreich1987how, benor1988completeness}.  But this is impractical due to its high polylogarithmic factor.

There are other MPC algorithms that focus on some restricted join types:  SSA \cite{asharov2023secure} presents $O(n \log n)$-cost algorithms for PK-FK joins and group-by operations under the three-server honest majority setting.  Senate \cite{senate} gives an $O(n \log n)$ algorithm for the even more restricted PK-PK joins.
FDJ \cite{fastdatabasejoin} reduces the complexity to linear, but it  only works for PK-PK joins.
Without the key constraints, these protocols must reveal the degree information, thus not satisfying the security definition of MPC.

\section{Preliminaries}

\subsection{Secret-Sharing} 
\label{sec:secret-share}
Secret sharing is a construction at the core of most MPC protocols.  An \textit{$(a,b)$-secret sharing scheme} splits an $\ell$-bit secret $v$ into $b$ shares, such that any $a-1$ of the shares reveal no information about $v$, while any $a$ shares allow the reconstruction of $v$.  In this paper, we adopt the following $(2,3)$-secret sharing scheme \cite{aby3}: Pick two random $\ell$-bit numbers $v_0, v_1$ independently, and set $v_2=v \oplus v_0 \oplus v_1$, where $\oplus$ denotes bit-wise XOR.  Then the three shares are $(v_0,v_1),(v_1,v_2)$, and $(v_2,v_0)$, respectively.  It should be clear that any two shares can reconstruct $v$, while any single share just consists of two independent random numbers, hence revealing nothing about $v$. 

\subsection{The Three-Party Model}\label{sec:model}
We now formally define the 3PC model under the context of query processing.  Let $\mathbf{R}$ be a database consisting of $k$ relations.  The schema of $\mathbf{R}$ is public but the data in each relation is privately owned by one party.  It is possible that some party does not own any data; in this case, it only takes part in the computation.
Given any query $Q$ on $\mathbf{R}$,
the three parties jointly execute a protocol to evaluate $Q$. 

The security of a protocol under 3PC is measured by its ability to defend against some adversaries. There are two types of adversaries: \textit{semi-honest} and \textit{malicious}. A \textit{semi-honest} adversary can view the transcript (i.e., all the messages sent and received during the protocol) of some parties, which are called \textit{corrupted parties}.  The security guarantee is that the adversary should not be able to infer anything about the data of the non-corrupted parties beyond the input and output size.  Following \cite{fastdatabasejoin, secrecy, secretsharedjoins,scape}, we only consider 3PC with honest majority, i.e., at most one party is corrupted.

We require the output of the protocol to be the secret shares of the query results.  The parties may reconstruct the query results if that is what has been agreed upon, or feed the results to some downstream applications.  More formally, this is captured by the following definition using the real-world ideal-world paradigm:

\begin{definition}
\label{def:security}
Let $x=(x_1,x_2,x_3)$ be the private inputs of the three parties where $x_i$ is the input of the $i$-th party for $i\in\{1,2,3\}$, and let $\mathrm{Real}_\pi(\kappa, C; x)$ be real-world view of the adversary, which consists of the transcript of protocol $\pi$ of the corrupted parties $C$ running with security parameter $\kappa$ on input $x$.
We say that $\pi$ is secure against a semi-honest adversary if there exists a (possibly randomized) function $\mathrm{Ideal}_Q(\kappa, C; x_C; n_1,n_2,n_3,m)$ such that, for any $C$ and $x_C$, $n_i=|x_i|$ ($i=1,2,3$), and $m=|Q(x)|$, the distributions of $\mathrm{Real}_\pi(\kappa, C, x)$ and $\mathrm{Ideal}_Q(\kappa, C; x_C; n_1,n_2,n_3,m)$ are indistinguishable in terms of $\kappa$, i.e., no polynomial-time (in $\kappa$) algorithm can distinguish them by probability more than $1/2 + \nu(\kappa)$ where $\nu(\kappa)$ is a negligible function of $\kappa$.
\end{definition}

This definition thus captures the intuition that the real-world view contains no information beyond the ideal-world view, which only depends on the data of the corrupted parties and the input/output size.
The semi-honest adversary only passively observes the corrupted parties.  On the other hand, a \textit{malicious} adversary has the additional power to make the corrupted parties deviate from the protocol.  In this paper, we only consider a semi-honest adversary. There are some standard techniques to harden a semi-honest protocol to defend against a malicious adversary \cite{chaum1984blind, goldreich1987how, scape, spdz}; we leave this extension to future work.

\paragraph{Cost model}
As with all prior works, we focus on the asymptotic behavior of the protocols in terms of the input size $n$ and output size $m$, while suppressing the dependency on the other parameters in the big-O notation.  This includes the query size (i.e., the number of relations and variables/attributes), the bit-length of the attributes, and the security parameter.  Nevertheless, when examining the concrete performance of a query plan, we do take these parameters into consideration.  When the running time and communication are asymptotically the same, we use the term ``cost'' to refer to both.
In addition to running time and communication cost, the number of communication rounds is also important.  We ensure that all the protocols in this paper run in $O(\log (n+m))$ rounds.

\subsection{Dummy Tuples}
For query processing on secret-shared data, we often need to pad dummy tuples to intermediate relations without letting the parties know whether a tuple is dummy or not.   This is done by adding a \textit{dummy marker} attribute to each intermediate relation.  The value of a dummy marker is either 0 (dummy) or 1 (not dummy), which is also secret-shared to the three parties.  In the description of the algorithms, we use $\bot$ to denote a dummy tuple. All dummy tuples are defined to be equal, while any non-dummy tuple is not equal to any dummy tuple.

\subsection{Three-Party Primitives}\label{sec:basic_func}

We will make use of the following primitives, where the input and output of each primitive are secret-shared with fresh randomness.
Since each primitive has been shown to be secure (i.e., satisfying Definition \ref{def:security}), any sequential composition of them is also secure \cite{security2000ran}.

\paragraph{Circuit-based computation}
There are 3PC protocols \cite{aby3} that, given the secret shares of two values $x$ and $y$, output $x\oplus y$ in secret-shared form (henceforth all input, intermediate data, and output are secret-shared unless specified otherwise), where $\oplus$ can be any standard arithmetic (addition, multiplication, comparison, etc.) or logical operation (AND, OR, XOR).  These protocols thus allow any circuit-based computation to be carried out in the 3PC model, with the cost proportional to the number of gates of the circuit and the number of rounds proportional to the depth.  In particular, we will make use of the following circuit-based primitive.

\paragraph{Prefix sum}
Given a binary associative operator $\oplus$, the prefix sum primitive takes an array $X=(x_1,\dots,x_n)$ as input and outputs the array $(x_1,x_1\oplus x_2,\dots,x_1\oplus x_2\oplus\dots\oplus x_n)$.  There is a classical $O(n)$-size, $O(\log n)$-depth circuit for this problem \cite{ladner1980parallel}.

We will also need a slightly more general version, called \textit{segmented prefix sum}.  It takes two arrays $A=(a_1,\dots,a_n)$ and $X=(x_1,\dots,x_n)$ as input, where equal elements in $A$ that appear consecutively define a segment.
The outputs are the prefix sums of each segment in $X$.
For example, if $A=(2,2,4,\bot,\bot)$, then the output is $(x_1,x_1\oplus x_2,x_3,x_4,x_4\oplus x_5)$. 
\cite{wang2022query} shows how to adapt the  prefix-sum circuit to compute segmented prefix sums.

There are two common choices for $\oplus$: When $\oplus$ is the arithmetic addition $+$ (define $x+\bot=\bot$),  
we simply call the output of (segmented) prefix sum as the \prefix{sum} of $X$ (segmented by $A$). The other common choice is 
\[
x \oplus y :=
\begin{cases}
    x, & \text{if } y = \bot; \\
    y, & \text{otherwise.} \\
\end{cases}
\]
This instantiation, which we call \prefix{copy}, replaces each dummy value in $X$ with the last non-dummy value before it (if there is one).

\paragraph{Permutation}
Permutation takes an array $X = (x_1, x_2, \cdots, x_n)$ and an array of permutation indices $P = (p_1, p_2, \cdots, p_n)$ (i.e., $P$ is a permutation of $[n] = (1, 2, \cdots, n)$), and outputs an array $Y = (y_1, y_2, \cdots, y_n)$, where $y_{p_i} = x_i$ for all $i\in [n]$, i.e., each $x_i$ is moved to position $p_i$ in the output. 
We call the output $Y$ the \texttt{permutation} of $X$ specified by $P$.
There is a 3PC permutation protocol with $O(n)$ cost and $O(1)$ rounds \cite{graphanalysis}. 

\begin{figure*}
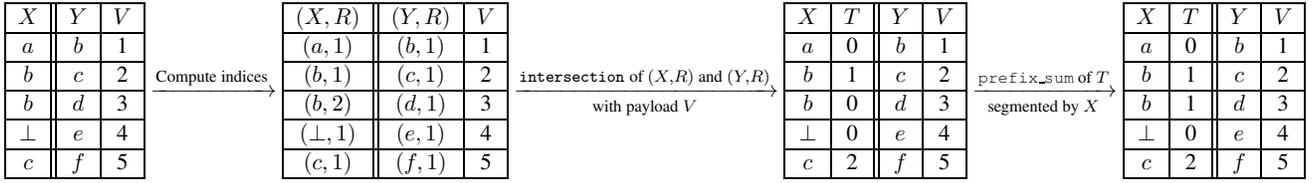

\resizebox{\textwidth}{!}{
\begin{tabular}{c}
\begin{tabular}{|c||c|c|}
\hline
$X$ & $Y$ & $V$ \\ \hline
$a$ & $b$ & 1  \\ \hline
$b$ & $c$ & 2  \\ \hline
$b$ & $d$ & 3  \\ \hline
$\perp$ & $e$ & 4  \\ \hline
$c$ & $f$ & 5  \\ \hline
\end{tabular}
$\xrightarrow{\text{Compute indices}}  $
\begin{tabular}{|c||c|c|}
\hline
$(X,R)$ & $(Y,R)$ & $V$ \\ \hline
$(a,1)$ &  $(b,1)$ & 1  \\ \hline
$(b,1)$ &  $(c,1)$ & 2  \\ \hline
$(b,2)$ &  $(d,1)$ & 3  \\ \hline
$(\perp,1)$ & $(e,1)$ & 4  \\ \hline
$(c,1)$ & $(f,1)$ & 5  \\ \hline
\end{tabular}
$\xrightarrow[\text{with payload } V]{\mathtt{intersection}\text{ of }(X,R)\text{ and }(Y,R)}  $
\begin{tabular}{|c|c||c|c|}
\hline
$X$ & $T$ & $Y$ & $V$ \\ \hline
$a$ & 0 & $b$ & 1  \\ \hline
$b$ & 1 & $c$ & 2  \\ \hline
$b$ & 0 & $d$ & 3  \\ \hline
$\perp$ & 0 & $e$ & 4  \\ \hline
$c$ & 2 & $f$ & 5  \\ \hline
\end{tabular}
$\xrightarrow[\text{segmented by }X]{\prefix{sum}\text{ of }T}  $
\begin{tabular}{|c|c||c|c|}
\hline
$X$ & $T$ & $Y$ & $V$ \\ \hline
$a$ & 0 & $b$ & 1  \\ \hline
$b$ & 1 & $c$ & 2  \\ \hline
$b$ & 1 & $d$ & 3  \\ \hline
$\perp$ & 0 & $e$ & 4  \\ \hline
$c$ & 2 & $f$ & 5  \\ \hline
\end{tabular}
\end{tabular}
}
\caption{Extended intersection example}
\label{fig:intersection}
\end{figure*}

\paragraph{Compaction}
This primitive takes two arrays $X=(x_1,\dots,x_n)$ and $T=(t_1,\dots,t_n)$ as input, where each $t_i\in\{0,1\}$. 
Each $x_i$ is said to be \textit{marked} if $t_i=1$. The \texttt{compaction} of $X$ on $T$ permutes $X$ such that all marked elements appear before non-marked elements, and the relative ordering of the marked elements must be preserved.  It is known that \texttt{compaction} can be done by \prefix{sum} and \texttt{permutation} \cite{efficient3pcsorting}, hence it also has $O(n)$ cost and $O(\log n)$ rounds. 
A useful case of \texttt{compaction} is to move all non-dummy tuples to be in the front of dummy tuples, where we use the dummy markers as $T$.


\paragraph{Intersection}
Intersection takes two arrays $X=(x_1,\dots,x_n)$ and $Y=(y_1,\dots,y_n)$ as input, where all non-dummy elements in $X$ are distinct and so are the non-dummy elements in $Y$.  It outputs an array $T=(t_1,\dots,t_n)$ such that $t_i=1$ indicates $x_i=y_j$ for some $j$, otherwise $t_i=0$. An exception is that for any dummy value $x_i$, it must have $t_i=0$.
A more general version of \texttt{intersection} supports payloads, in which each $y_j$ is also attached with a payload $v_j$. In the output, $t_i=v_j$ if $x_i=y_j$ for some $j$ (note that there is at most one $j$ such that $x_i=y_j$), otherwise $t_i=0$. Similarly, if $x_i$ is dummy then it must have $t_i=0$.  \cite{fastdatabasejoin} provide a 3PC \texttt{intersection} protocol with payload that has $O(n)$ cost and $O(1)$ rounds.  Their protocol does not directly support dummy elements, but they can be first changed to something that is guaranteed not to intersect. For example, we can change each $x_i$ (resp.\ $y_i$) to $(x_i,0)$ (resp.\ $(y_i,0)$) if it is not dummy, and to $(i,1)$ (resp.\ $(i+n,1)$) if it is.  This can be done by a circuit of size $O(n)$ and depth $O(1)$.

\paragraph{Extended intersection}
For our purpose, we will need a further generalization of \texttt{intersection}, where elements in $X$ are not necessarily distinct, but equal elements must be consecutive.  We still require the elements in $Y$ to be distinct so that the output is still well-defined. 
Here is our solution:  We first assign consecutive indices to each segment of equal elements in $X$.  This can be done by computing the \prefix{sum} of $(1,1,\dots,1)$ segmented by $X$.  Let $r_i$ be the index of $x_i$.  Next, we compute the \texttt{intersection} of $((x_1,r_1),\dots,(x_n,r_n))$ and $((y_1,1),\dots,(y_n,1))$, possibly with payload. For each segment in $X$, only the first element (i.e., the one with $r_i=1$) gets the correct payload, while other elements get 0. Finally, we use segmented \prefix{sum} to copy the payload of the first element in each segment to the other elements.  Please see Figure~\ref{fig:intersection} for an example.
Thus the protocol has $O(n)$ cost and $O(\log n)$ rounds.

\paragraph{Expansion}
Expansion takes two arrays $X=(x_1,\dots,x_n)$ and $D=(d_1,\dots,d_n)$ as input, together with a parameter $m$.  It is guaranteed that (1) $x_i=\bot$ implies $d_i=0$; and (2) $D_{\Sigma}:=\sum_{i=1}^n d_i \le m$. Each $d_i$, which we call the \textit{degree} of $x_i$, is a non-negative integer (hence non-dummy) that indicates the number of repetitions that $x_i$ should appear in the output. The output is a length-$m$ array:
\[(\underbrace{x_1,\dots,x_1}_{d_1\text{ times}},\underbrace{x_2,\dots,x_2}_{d_2\text{ times}},\dots,\underbrace{x_n,\dots,x_n}_{d_n\text{ times}},\underbrace{\bot,\dots,\bot}_{m-D_{\Sigma}\text{ times}}).\]
We call the output array the \texttt{expansion} of $X$ with degree $D$ and output size $m$. \cite{scape} show how this can be done with $O(n+m)$ cost and $O(\log(n+m))$ rounds under the constraints that $m=D_{\Sigma}$ and $d_i\neq 0$ for all $i$ (hence $X$ contains no dummies).  Their algorithm works as follows:
\begin{enumerate}
    \item Compute the first index where each $x_i$ will appear in the output array. This can be done by computing the \prefix{sum} of $\{1,d_1,d_2,\dots,d_{n-1}\}$. Denote the result as $\{p_1,\dots,p_n\}$.
    \item Distribute each $x_i$ to index $p_i$.  Let $Y=\{y_1,\dots,y_m\}$ be the resulting array.  Any $y_j$ such that $j\ne p_i$ for all $i$ is dummy.
    \item Run \prefix{copy} on $Y$.
\end{enumerate}

Below we show how to remove the two constraints above. To remove the constraint $m=D_{\Sigma}$, we just need to set those elements in $Y$ with positions larger than $D_{\Sigma}$ to dummy after the final step, which can be done by a circuit of size $O(m)$ and depth $O(1)$. Removing the $d_i\ne 0$ constraint is trickier, since $d_i=0$ means that $x_i$ and $x_{i+1}$ have the same target index, and step (2) above will fail. Our solution is that, after (1), we set $p_i$ to $m+i$ if $d_i=0$.  Note that this guarantees that all $p_i$'s are distinct.  Then we run step (2) to output a length-$(m+n)$ array, where those $x_i$ with $d_i=0$ will be moved to index $p_i>m$. After step (2), we discard the last $n$ elements, and run step (3). It should be clear that this modified algorithm still has  $O(n+m)$ cost and $O(\log(n+m))$ rounds.

\section{Relational Operators}
In this section, we first introduce \textit{consistent sort} and the corresponding \textit{ranks}.  Then we show how we support various relational operators in linear complexity using these ranks.
Our protocol descriptions include conditional data assignment \textbf{if-then-else} statements. When the condition is secret-shared, this assignment is realized by a series of comparison and multiplexers on the secret shares. We keep such statements for simplicity, without explicitly discussing their concrete implementations.

\subsection{Consistent Sort and Ranks}\label{sec:rank}
We introduce the first step of LINQ in this section, which is the parties performing consistent sort on the relations they own in plaintext.
Let $\bm{E}$ be a possible join key or group-by key, which may consist of one or multiple attributes.  Let $\mathbf{dom}(\bm{E})$ be the domain of $\bm{E}$.  We say that the relations have been \textit{consistently sorted} on $\bm{E}$, if (1) tuples with the same values on $\bm{E}$ in each relation are all consecutive, and (2) for any two values $x,y \in \mathbf{dom}(\bm{E}), x\ne y$, their relative ordering in all the relations are the same (if they both appear).  For example, consider two relations $R_1(A,B), R_2(B,C)$ with join key $\bm{E}=B$.  Then $R_1=((a_1,3), (a_2, 1), (a_3, 1), (a_4, 2))$ and $R_2=((4,c_1), (3,c_2), (3,c_3), (1, c_4))$ are consistently sorted on $B$ (the ordering is 4,3,1,2), but $R_1$ and $R'_2=((1, c_1), (3, c_2), (3, c_3),(4,c_4))$ are not.  Our observation is that, to compute $R_1\Join R_2$ using sort-merge-join, a consistent sort suffices.

Sorting all relations by $\bm{E}$ produces a consistent sort, but this takes $O(n\log n)$ time.  Below we show that a consistent sort can be actually done in $O(n)$ time due to its more relaxed requirement. 
Let $h:\mathbf{dom}(\bm{E})\rightarrow [n]$ be a public pairwise-independent hash function.  For each relation $R(\bm{F})$ such that $\bm{E} \subseteq \bm{F}$,  we first hash the tuples to $n$ buckets using $h$, where the $i$-th bucket is $H_i := \{t\in R \mid h(t.\bm{E}) = i\}$ for  $i\in[n]$.  Then we sort each $H_i$ by the natural ordering of $\bm{E}$ and output $H_1,H_2,\dots,H_n$ in order.  To sort each $H_i$, we take the advantage that the tuples in $H_i$ have only $O(1)$ distinct $\bm{E}$ values in expectation.  So we first sort the distinct $\bm{E}$ values, count how many tuples each distinct $\bm{E}$ value has, and compute their prefix sums.  These prefix sums indicate the starting position for each group of tuples having the same value on $\bm{E}$.  We can then put all the tuples in $H_i$ in the order of $\bm{E}$ in one pass. 

\begin{example}
    Suppose we do a consistent sort of $R(A, B) = ((a_1, 2)$, $(a_2, 1), (a_3, 3), (a_4, 1),(a_5, 4),(a_6, 2))$ on $\bm{E} = B$.  The hash function has $h(1) = 3$, $h(2) = 3$, $h(3) = 1$, $h(4)=3$.  The hash buckets are therefore $H_1 = ((a_3, 3))$, $H_3 = ((a_1, 2), (a_2, 1), (a_4, 1), (a_5, 4), (a_6, 2))$, $H_2=H_4=\emptyset$.

    Then we sort each $H_i$. Here we only consider $H_3$, which contains 3 distinct $\bm{E}$ values: $1,2,4$, and they have 2, 2, and 1 tuples, respectively.  The prefix sum yields the starting positions: $p_1 = 0$, $p_2 = 2$, $p_4=4$ (assuming the output array is 0-indexed).  Then we scan the tuples in $H_3$ and put them into the right locations: $(a_1, 2)$ is assigned to location $p_2 = 2$, and we increment $p_2$; $(a_2, 1)$ is assigned to location $p_1 = 0$, and we increment $p_1$; $(a_4, 1)$ is  assigned to location $p_1 = 1$, and we increment $p_1$; etc. 
    So the sorted $H_3 = ((a_2, 1), (a_4, 1), (a_1, 2), (a_6, 2), (a_5, 4))$.

    Finally, we concatenate all the $H_i$ to get the consistently sorted $R=((a_3, 3), (a_2, 1), (a_4, 1), (a_1, 2), (a_6, 2), (a_5, 4))$.
\end{example}

It should be clear that this algorithm produces a consistent sort, since the ordering is first decided by $h$, and then ties are broken by the natural order of $\bm{E}$.  We analyze its running time in the following lemma:

\begin{lemma}
    \label{lem}
    The algorithm above takes $O(n)$ time in expectation to consistently sort a relation by $\bm{E}$. 
\end{lemma}
\begin{proof}
It should be clear that hashing the tuples to the $n$ buckets, counting the frequencies of each distinct $\bm{E}$ value in each bucket, computing the prefix sums, and the final pass to put items in the right locations all take linear time.  The only possibly non-linear step is sorting the distinct $\bm{E}$ values in each bucket, which we analyze below.

Without loss of generality, consider the worst-case scenario where all $n$ tuples have distinct $\bm{E}$ values, denoted as $x_1, x_2, \cdots, x_n$. Let $X$ be the number of tuples hashed to $H_1$, and $X_i$ be the indicator that $x_i$ is hashed to $H_1$, i.e., $X_i=1$ if $h(x_i)=1$, otherwise $X_i=0$. 
Then $\E[X_i]=\Pr[X_i=1]=1/n$ and $X=\sum_{i=1}^n X_i$. Since $h$ is a pairwise-independent hash function, we have $\E[X_iX_j]=\Pr[X_i=1,X_j=1]=\Pr[X_i=1]\cdot \Pr[X_j=1]=1/n^2$ for any $i\neq j$. Then
\begin{align*}
    \E[X^2]&=\E\left[\sum_{i=1}^n X_i^2+\sum_{i\neq j}X_iX_j\right]=\sum_{i=1}^n \E[X_i]+\sum_{i\neq j}\E[X_iX_j]\\
    &=n\cdot \frac{1}{n}+n(n-1)\cdot \frac{1}{n^2}<2.
\end{align*}

Therefore, the cost of sorting $H_1$ (by bubble sort) is expected to be $\E[X^2]=O(1)$. By symmetric and linearity of expectation, the total cost of sorting all the buckets is expected to be $O(n)$.
\qedhere
\end{proof}

After consistently sorting $R(\bm{F})$ by $\bm{E}$, we add a new attribute $\bmrank{E}$, which we call the \textit{rank} of $\bm{E}$, to $R$. To differentiate from original attributes, we refer to it as a \textit{rank attribute}. For any $t\in R$, $t.\rank{\bm{E}}$ is simply the index of $t\in R$ after the sort.  If there are dummy tuples in $R$, we define $h(\bot)=n$ and $t.\bm{E}< \bot.\bm{E}$ for any $t\neq\bot$, so that \textit{the ranks of dummy tuples are always larger than that of non-dummy tuples}.  

Let $\mathcal{E}$ be the set of all join and group-by keys.  For each $R(\bm{F})$, the party computes a rank attribute $\bmrank{E}$ for each $\bm{E} \in \mathcal{E}$ by consistent sort. 
The total running time is still $O(n)$, as we take the number of relations and attributes as constant.  We use the notation $R(\bm{F};\rank{\mathcal{E}})$, where $\rank{\mathcal{E}}=\{\rank{\bm{E}}\mid \bm{E}\subseteq \bm{F}, \bm{E}\in\mathcal{E}\}$ to denote a relation $R$ augmented with the ranks $\rank{\mathcal{E}}$. 

\subsection{Basic Relational Operators}
Ranks, as we will see below, help to reduce the query processing cost to linear.  Meanwhile, as we consider a query plan consisting of multiple relational operators, these ranks must remain valid after each operator.

For example, a join operator could filter a tuple out or make multiple copies of a tuple, which turns the original ranks to be invalid. It requires non-trivial recomputation to the ranks so that they still form a permutation of $[n]$ and it corresponds to the consistent sort, i.e., they are always the indices of the tuples after sorting by $h$ and then by $\bm{E}$. 

\paragraph{Selection}
The selection operator $\sigma_{\gamma}(R)$ returns the subset of tuples of $R$ that pass the predicate $\gamma$, i.e., $\sigma_{\gamma}(R):=\{t\in R\mid \gamma(t)=\mathtt{true}\}$.  We assume that $\gamma$ can be evaluated by a circuit of $O(1)$ size, so it takes $O(n)$ cost and $O(1)$ rounds to evaluate $\gamma$ on all tuples of $R$.  However, since the results of the predicates are also secret-shared, the parties cannot and should not remove them physically.  Instead, we set a tuple to dummy if it does not pass the predicate.  

Since some tuples are turned into dummies after selection, the ranks are no longer valid.  The parties cannot afford to recompute them as the consistent sort algorithm above can only be done in plaintext.  Thankfully, since $\sigma_\gamma(R)$ is a subset of $R$, we can update the ranks by first  permuting the tuples by the old ranks and then compacting out the dummy tuples, as shown in Algorithm~\ref{alg:update_rank}.

\begin{algorithm}
    \caption{Update ranks}
    \label{alg:update_rank}
    \KwIn{Relation $R(\bm{F}; \rank{\mathcal{E}})$}
    \KwOut{Relation $R(\bm{F}; \rank{\mathcal{E}})$ with updated ranks}
    \For{$\bmrank{E} \in \rank{\mathcal{E}}$}{
        $R\gets$ \texttt{permutation} of $R$ by $R.\bmrank{E}$\;
        $R\gets$ \texttt{compaction} of $R$\;
        \ParFor{$i\gets 1$ \KwTo $n$}{
            $R[i].\bmrank{E}\gets i$\;
        }
    }
    \KwRet{$R$}
\end{algorithm}

\paragraph{Group-by-aggregation}
Let $\pi_{\bm{E}}^\oplus$ be the group-by-aggregation operator with  group-by attributes $\bm{E}$ and aggregate function $\oplus$, 
which can be \texttt{count}, \texttt{sum}, \texttt{max}, or \texttt{min}.
We first permute $R$ by $\bmrank{E}$. Then tuples that have the same value on $\bm{E}$ are consecutive.  Then we can compute the aggregation by evaluating the prefix sum circuit segmented by $\bm{E}$, using an $\oplus$ corresponding to the aggregate function (e.g., $+$ for \texttt{sum}). This way, the last tuple of each segment holds the aggregate for the group, so we set the other tuples to dummy.  This step can be done by a circuit of size $O(n)$ and depth $O(1)$.  
Note that this operator also returns a subset of tuples (one from each group) from the input relation, so the ranks can be updated as before.  Please see Algorithm~\ref{alg:gba} for the complete algorithm for group-by-aggregation. 

\begin{algorithm}
    \caption{Group-by-sum protocol $\pi_{\bm{E}}^{\mathtt{sum}(A)}$}
    \label{alg:gba}
    \KwIn{Relation $R(\bm{F}; \rank{\mathcal{E}})$ with public size $n$}
    \KwOut{Group-by-sum relation $T(\bm{E}, \mathtt{sum}(A); \rank{\mathcal{E}'})$}
    $T \gets$ \texttt{permutation} of $R$ by $\bmrank{E}$\;
    Add attribute $\mathtt{sum}(A)$ to $T$ with $T.\mathtt{sum}(A) \gets$ \prefix{sum} of $T.A$ segmented by $T.\bm{F}$\;
    \ParFor{$i \gets 1$ \KwTo $n-1$}{
        \If{$T[i].\bm{E} = T[i+1].\bm{E}$}{
            $T[i] \gets \bot$\;
        }
    }
    Update rank attributes of $T$ by Algorithm~\ref{alg:update_rank}\;
    \KwRet{$T$}
\end{algorithm}

\paragraph{Projection}
For projection $\pi_{\bm{F'}}(R(\bm{F},\rank{\mathcal{E}}))$ where $\bm{F'}\subset \bm{F}$, 
we simply drop the attributes $\bm{F}-\bm{F'}$.  Meanwhile, we can also drop any rank attribute $\rank{\bm{E}}$ if $\bm{E} \not\subseteq \bm{F}'$.  There is no need to update the ranks.


Note that a distinct projection on $\bm{F'}$ should be done as a group-by-aggregation using $\bm{F'}$ as the group-by attribute, but the aggregation function is irrelevant. 

\paragraph{Semi-join}
The semi-join operator takes two relations $R(\bm{F_R};\rank{\mathcal{E}}_R)$ and $S(\bm{F}_S;\rank{\mathcal{E}}_S)$ as input, and returns the set of tuples in $R$ that can join with $S$, i.e., $R\ltimes S=\{t_1\in R\mid \exists\, t_2\in S:t_2.\bm{E}=t_1.\bm{E}\}$, where $\bm{E}=\bm{F_R}\cap \bm{F_S}$ is the join key. Similar to the selection operator, the tuples in $R$ that do not join with $S$ are not removed but set to dummy.

The idea is to first permute $R$ by $R.\bmrank{E}$ and permute $S$ by $S.\bmrank{E}$.  Then for all tuples in $S$ with the same value on $S.\bm{E}$, we set all but one to dummy.  This is done as in Line 3--5 of Algorithm~\ref{alg:gba}.  Then we compute the extended \texttt{intersection} of $R.\bm{E}$ and $S.\bm{E}$. 
Finally, we set tuples that get a zero payload to dummy and update the ranks by Algorithm~\ref{alg:update_rank}.

\subsection{Join}\label{sec:join}
It remains to show how to perform a join. Recall that a (natural) join takes two relations $R(\bm{F}_R;\rank{\mathcal{E}}_R)$ and $S(\bm{F}_S;\rank{\mathcal{E}}_S)$ as input, where $\rank{\mathcal{E}}_R$ and $\rank{\mathcal{E}}_S$ denote the rank attributes, and returns all combinations of tuples from the two relations such that they have the same values on their join key, i.e., the join result $T:=R\Join S$ is
\[T(\bm{F}_R\cup\bm{F}_S;  \rank{\mathcal{E}}_R\cup\rank{\mathcal{E}}_S)=\{(t_R,t_S)\mid t_R\in R, t_S\in S,t_R.\bm{E}=t_S.\bm{E}\},\] where the join key $\bm{E}=\bm{F}_R\cap \bm{F}_S$ is in $\mathcal{E}_R\cap\mathcal{E}_S$.  Note that the join result contains all the rank attributes $\rank{\mathcal{E}}_R\cup\rank{\mathcal{E}}_S$, so as to support multi-way joins, which will be discussed in Section \ref{sec:freeconnex}. 

Assume for simplicity that $|R|=|S|=n$.  We also assume for now that the join size $m \geq |R\Join S|$ is given. The last $m - |R \Join S|$ tuples of the join result $T$ are dummy. The idea follows \cite{krastnikov2020efficient,scape,secretsharedjoins}, but due to the rank attributes, we are able to implement it with linear complexity. For any tuple $t_R\in R$, note that the number of repetitions of $t_R$ in $R\Join S$ is exactly how many times $t_R.\bm{E}$ appears in $S.\bm{E}$. Therefore, we first count the frequencies of the values in $S.\bm{E}$, which can be done by group-by-count $\pi_{\bm{E}}^{\mathtt{count}}$. Then we permute $R$ by $\bm{E}$, and attach these frequencies to the corresponding tuples in $R$, which can be realized by  extended \texttt{intersection}. Next, we compute the \texttt{expansion} of $R$ according to these frequencies. In a symmetric way we also expand $S$. 
Then we permute $S$ by some specific permutation indices (i.e., $G$ in Algorithm~\ref{alg:join}) so as to align each tuple in $S$ with the corresponding joined tuple in $R$.
Finally, zipping the two expanded relations together yields the join results.  See Algorithm~\ref{alg:join} for details and Figure~\ref{fig:join} for an example.

\begin{algorithm}[h]
    \caption{Join protocol $\Join$}
    \label{alg:join}
    \KwIn{Relation $R(\bm{F_R}; \rank{\mathcal{E}}_R)$ and $S(\bm{F_S}; \rank{\mathcal{E}}_S)$ each with public size $n$}
    \KwOut{Join result $T(\bm{F_R}\cup\bm{F_S};  \rank{\mathcal{E}}_R\cup\rank{\mathcal{E}}_S)$ with public size $m$}
    \tcp{Use permutation to sort relations by join key.}
    $R \gets$ \texttt{permutation} of $R$ by $R.\bm{E}$\; 
    $S \gets$ \texttt{permutation} of $S$ by $S.\bm{E}$\;
    \tcp{Calculate the appearance of each tuple in $R$ and $S$, and expand the tuple by the appearance.}
    $S'{(\bm{E}, \mathtt{count})} \gets \pi_{\bm{E}}^{\mathtt{count}}(S)$\;
    $R.D_S\gets$ extended \texttt{intersection} of $R.\bm{E}$ and $S'.\bm{E}$ with payload $S'.\mathtt{count}$\;
    $R'' \gets$ \texttt{expansion} of $R$ by degree $R.D_S$ with output size $m$ (by Algorithm~\ref{alg:expansion_rank})\;
    $R'{(\bm{E}, \mathtt{count})} \gets \pi_{\bm{E}}^{\mathtt{count}}(R)$\;
    $S.D_R\gets$ extended \texttt{intersection} of $S.\bm{E}$ and $R'.\bm{E}$ with payload $R'.\mathtt{count}$\;
    $S.D_S\gets$ extended \texttt{intersection} of $S.\bm{E}$ and $S'.\bm{E}$ with payload $S'.\mathtt{count}$\;
    $S.O\gets (1,2,\dots,n)$\;
    $S.J\gets$ \prefix{sum} of $(1,\dots,1)$ segmented by $S.\bm{E}$\;
    $S'' \gets$ \texttt{expansion} of $S$ by degree $S.D_R$ with output size $m$ (by Algorithm~\ref{alg:expansion_rank})\; 
    \tcp{Rearrange the expanded relation $S$.}
    $S''.I\gets$ \prefix{sum} of $(1,\dots,1)$ segmented by $S''.O$\;
    Add an attribute $G$ to $S''$\;
    \ParFor{$i \gets 1$ \KwTo $m$} {
        $t \gets S''[i]$\;
        $S''[i].G \gets i+(t.I-1)\cdot t.D_S+t.J-(t.J-1)\cdot t.D_R-t.I$\;
    }
    $S''\gets$ \texttt{permutation} of $S''$ by $G$\;
    \tcp{Keep $S''.\bmrank{E}$ consistent with $R''.\bmrank{E}$.}
    $S''.\bmrank{E}\gets (1,\dots,m)$\;
    $T \gets (R''.\bm{F_R}, S''.\bm{F_S}; R''.\rank{\mathcal{F}_R}\cup S''.\rank{\mathcal{F}_S})$\;
    \KwRet{$T$}
\end{algorithm}

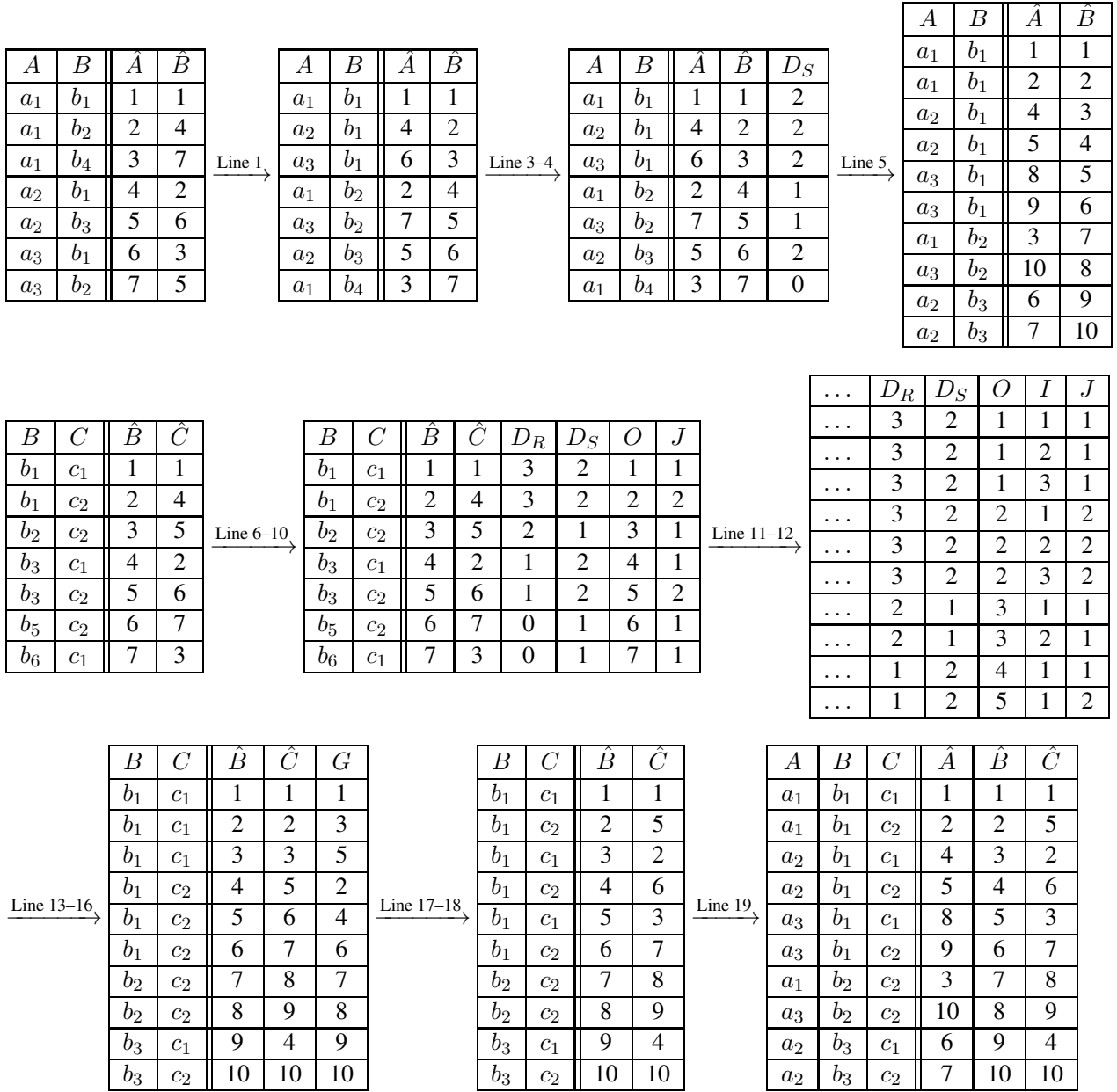
\begin{figure*}
\resizebox{\textwidth}{!}{
\begin{tabular}{l}
\begin{tabular}{|c|c||c|c|}
\hline
${A}$ & ${B}$ & ${\rank{A}}$ & ${\rank{B}}$ \\ \hline
$a_1$ & $b_1$ & 1 & 1 \\ \hline
$a_1$ & $b_2$ & 2 & 4 \\ \hline
$a_1$ & $b_4$ & 3 & 7 \\ \hline
$a_2$ & $b_1$ & 4 & 2 \\ \hline
$a_2$ & $b_3$ & 5 & 6 \\ \hline
$a_3$ & $b_1$ & 6 & 3 \\ \hline
$a_3$ & $b_2$ & 7 & 5 \\ \hline
\end{tabular}
$\xrightarrow{\text{Line 1}}$
\begin{tabular}{|c|c||c|c|}
\hline
${A}$ & ${B}$ & ${\rank{A}}$ & ${\rank{B}}$ \\ \hline
$a_1$ & $b_1$ & 1 & 1 \\ \hline
$a_2$ & $b_1$ & 4 & 2 \\ \hline
$a_3$ & $b_1$ & 6 & 3 \\ \hline
$a_1$ & $b_2$ & 2 & 4 \\ \hline
$a_3$ & $b_2$ & 7 & 5 \\ \hline
$a_2$ & $b_3$ & 5 & 6 \\ \hline
$a_1$ & $b_4$ & 3 & 7 \\ \hline
\end{tabular}
$\xrightarrow{\text{Line 3--4}}$
\begin{tabular}{|c|c||c|c|c|}
\hline
${A}$ & ${B}$ & ${\rank{A}}$ & ${\rank{B}}$ & $D_S$ \\ \hline
$a_1$ & $b_1$ & 1 & 1 & 2 \\ \hline
$a_2$ & $b_1$ & 4 & 2 & 2 \\ \hline
$a_3$ & $b_1$ & 6 & 3 & 2\\ \hline
$a_1$ & $b_2$ & 2 & 4 & 1\\ \hline
$a_3$ & $b_2$ & 7 & 5 & 1\\ \hline
$a_2$ & $b_3$ & 5 & 6 & 2\\ \hline
$a_1$ & $b_4$ & 3 & 7 & 0\\ \hline
\end{tabular}
$\xrightarrow{\text{Line 5}}$
\begin{tabular}{|c|c||c|c|}
\hline
${A}$ & ${B}$ & ${\rank{A}}$ & ${\rank{B}}$  \\ \hline
$a_1$ & $b_1$ & 1 & 1 \\ \hline
$a_1$ & $b_1$ & 2 & 2 \\ \hline
$a_2$ & $b_1$ & 4 & 3 \\ \hline
$a_2$ & $b_1$ & 5 & 4 \\ \hline
$a_3$ & $b_1$ & 8 & 5 \\ \hline
$a_3$ & $b_1$ & 9 & 6 \\ \hline
$a_1$ & $b_2$ & 3 & 7 \\ \hline
$a_3$ & $b_2$ & 10 & 8 \\ \hline
$a_2$ & $b_3$ & 6 & 9 \\ \hline
$a_2$ & $b_3$ & 7 & 10 \\ \hline
\end{tabular}
\multirow{3}{1em}{
\begin{tikzpicture}
\coordinate  (c0) at(0,0) ;
\coordinate  (c1) at(0.5,0) ;
\coordinate  (c2) at(0.5,-11.5) ;
\coordinate  (c3) at(0,-11.5) ;
\draw[-] (c0) -- (c1);
\draw[-] (c1) -- (c2);
\draw[->] (c2) -- (c3) ;
\end{tikzpicture}
}
\bigskip\\
\begin{tabular}{|c|c||c|c|}
\hline
${B}$ & ${C}$ & ${\rank{B}}$ & ${\rank{C}}$ \\ \hline
$b_1$ & $c_1$ & 1 & 1 \\ \hline
$b_1$ & $c_2$ & 2 & 4 \\ \hline
$b_2$ & $c_2$ & 3 & 5 \\ \hline
$b_3$ & $c_1$ & 4 & 2 \\ \hline
$b_3$ & $c_2$ & 5 & 6 \\ \hline
$b_5$ & $c_2$ & 6 & 7 \\ \hline
$b_6$ & $c_1$ & 7 & 3 \\ \hline
\end{tabular}
$\xrightarrow{\text{Line 6--10}}$
\begin{tabular}{|c|c||c|c|c|c|c|c|}
\hline
${B}$ & ${C}$ & ${\rank{B}}$ & ${\rank{C}}$ & $\!D_R\!$ & $\!D_S\!$ & $O$ & $J$\\ \hline
$b_1$ & $c_1$ & 1 & 1 & 3 & 2 & 1 & 1\\ \hline
$b_1$ & $c_2$ & 2 & 4 & 3 & 2 & 2 & 2\\ \hline
$b_2$ & $c_2$ & 3 & 5 & 2 & 1 & 3 & 1 \\ \hline
$b_3$ & $c_1$ & 4 & 2 & 1 & 2 & 4 & 1 \\ \hline
$b_3$ & $c_2$ & 5 & 6 & 1 & 2 & 5 & 2 \\ \hline
$b_5$ & $c_2$ & 6 & 7 & 0 & 1 & 6 & 1 \\ \hline
$b_6$ & $c_1$ & 7 & 3 & 0 & 1 & 7 & 1 \\ \hline
\end{tabular}
$\xrightarrow{\text{Line 11--12}}$
\begin{tabular}{|c|c|c|c|c|c|}
\hline
$\dots$ & $\!D_R\!$ & $\!D_S\!$ & $O$ & $I$ & $J$\\ \hline
$\dots$ & 3 & 2 & 1 & 1 & 1\\ \hline
$\dots$ & 3 & 2 & 1 & 2 & 1\\ \hline
$\dots$ & 3 & 2 & 1 & 3 & 1\\ \hline
$\dots$ & 3 & 2 & 2 & 1 & 2\\ \hline
$\dots$ & 3 & 2 & 2 & 2 & 2\\ \hline
$\dots$ & 3 & 2 & 2 & 3 & 2\\ \hline
$\dots$ & 2 & 1 & 3 & 1 & 1 \\ \hline
$\dots$ & 2 & 1 & 3 & 2 & 1 \\ \hline
$\dots$ & 1 & 2 & 4 & 1 & 1 \\ \hline
$\dots$ & 1 & 2 & 5 & 1 & 2 \\ \hline
\end{tabular}\bigskip \\
$\xrightarrow{\text{Line 13--16}}$
\begin{tabular}{|c|c||c|c|c|}
\hline
$B$ & $C$ & $\rank{B}$ & $\rank{C}$ & $G$ \\\hline
$b_1$ & $c_1$ & 1 & 1 & 1 \\ \hline
$b_1$ & $c_1$ & 2 & 2 & 3 \\ \hline
$b_1$ & $c_1$ & 3 & 3 & 5 \\ \hline
$b_1$ & $c_2$ & 4 & 5 & 2 \\ \hline
$b_1$ & $c_2$ & 5 & 6 & 4 \\ \hline
$b_1$ & $c_2$ & 6 & 7 & 6 \\ \hline
$b_2$ & $c_2$ & 7 & 8 & 7 \\ \hline
$b_2$ & $c_2$ & 8 & 9 & 8 \\ \hline
$b_3$ & $c_1$ & 9 & 4 & 9 \\ \hline
$b_3$ & $c_2$ & 10 & 10 & 10 \\ \hline
\end{tabular}
$\xrightarrow{\text{Line 17--18}}$
\begin{tabular}{|c|c||c|c|}
\hline
$B$ & $C$ & $\rank{B}$ & $\rank{C}$ \\\hline
$b_1$ & $c_1$ & 1 & 1  \\ \hline
$b_1$ & $c_2$ & 2 & 5  \\ \hline
$b_1$ & $c_1$ & 3 & 2  \\ \hline
$b_1$ & $c_2$ & 4 & 6  \\ \hline
$b_1$ & $c_1$ & 5 & 3  \\ \hline
$b_1$ & $c_2$ & 6 & 7  \\ \hline
$b_2$ & $c_2$ & 7 & 8  \\ \hline
$b_2$ & $c_2$ & 8 & 9  \\ \hline
$b_3$ & $c_1$ & 9 & 4  \\ \hline
$b_3$ & $c_2$ & 10 & 10 \\ \hline
\end{tabular}
$\xrightarrow{\text{Line 19}}$
\begin{tabular}{|c|c|c||c|c|c|}
\hline
${A}$ & $B$ & $C$ & $\rank{A}$ & $\rank{B}$ & $\rank{C}$ \\\hline
$a_1$ & $b_1$ & $c_1$ & 1 & 1 & 1  \\ \hline
$a_1$ & $b_1$ & $c_2$ & 2 & 2 & 5  \\ \hline
$a_2$ & $b_1$ & $c_1$ & 4 & 3 & 2  \\ \hline
$a_2$ & $b_1$ & $c_2$ & 5 & 4 & 6  \\ \hline
$a_3$ & $b_1$ & $c_1$ & 8 & 5 & 3  \\ \hline
$a_3$ & $b_1$ & $c_2$ & 9 & 6 & 7  \\ \hline
$a_1$ & $b_2$ & $c_2$ & 3 & 7 & 8  \\ \hline
$a_3$ & $b_2$ & $c_2$ & 10& 8 & 9  \\ \hline
$a_2$ & $b_3$ & $c_1$ & 6 & 9 & 4  \\ \hline
$a_2$ & $b_3$ & $c_2$ & 7 & 10 & 10 \\ \hline
\end{tabular}
\end{tabular}}
\caption{An example of join operator with $n=7$, $m=10$, $\bm{F_R}=\{A,B\}$, $\mathcal{E}_R=\{(A),(B)\}$, $\bm{F_S}=\{B,C\}$, and $\mathcal{E}_R=\{(B),(C)\}$}
\label{fig:join}
\end{figure*}

\paragraph{Expansion with ranks} 
In order to support multi-way joins still with linear complexity, we must recompute the rank attributes of $T$.  Observing that the only operation in the join algorithm that introduces new tuples is \texttt{expansion}, we just need to show how to recompute the ranks after \texttt{expansion}.

Let the input be $R(\bm{F}, D;\rank{\mathcal{E}})$ with size $n$, where $R.D$ is the degree for expansion, with the promise that $t.D=0$ if $t$ is dummy.  Recall that the output of expansion is a relation $T(\bm{F}; \rank{\mathcal{E}})$ of size $m$, where for each $t\in R$, $t.\bm{F}$ appears $t.D$ times in $T.\bm{F}$. After the expansion, we need to recompute each rank attribute $\rank{\bm{E}} \in \rank{\mathcal{E}}$.
A natural idea is to permute $R$ by $\bmrank{E}$, compute the \texttt{expansion} of it, and then set $T[i].\bmrank{E}$ to $i$ for all $i\in[m]$. However, this result is a permutation of the expansion of the original relation, and it is not clear how to permute it to the correct order. Moreover, it does not work if $|\mathcal{E}|\ge 2$, as we could only permute and expand in a specific order.

Our solution is to analyze how the ranks change after expansion: For any $t\in R$, the ranks of the $t.D$ repetitions of $t$ in $T$ should range from $s+1$ to $s+t.D$, where $s$ is the sum of degrees of tuples with ranks less than $t$ in $R$. Therefore, we first add an attribute $O$ to $R$ to record the initial order of the relation. Then for each $\bm{E}\in\mathcal{E}$, we permute $R$ by $\bmrank{E}$, compute the \prefix{sum} of $R.D$, and then update the rank attribute $\bmrank{E}$. After all the rank attributes are updated, we recover the relation to its original order by permuting it by $R.O$, and then we can compute the \texttt{expansion} as usual to get $T$. The ranks of $T$ in each repetition are then added by its local rank in each segment, which can also be computed by a segmented \prefix{sum}. Recall that the original \texttt{expansion} supports the case $m>D_\Sigma$, which results in $m-D_\Sigma$ dummy tuples at the end of $T$, and their ranks are not computed correctly. We update their ranks to the current orders in $T$. See Algorithm~\ref{alg:expansion_rank} and Figure~\ref{fig:expansion_rank} for details. The cost and the number of rounds are $O(n+m)$ and $O(\log (n+m))$ respectively.

\begin{algorithm}[h]
    \caption{Expansion with ranks protocol}
    \label{alg:expansion_rank}
    \KwIn{Relation $R(\bm{F}, D; \rank{\mathcal{E}})$ with public size $n$}
    \KwOut{$T(\bm{F}; \rank{\mathcal{E}})$ with public size $m$}
    Add an attribute $O$ to $R$ with $R[i].O=i$ for all $i\in[n]$\;
    \For{$\bm{E} \in \mathcal{E}$}{
        $R \gets $ \texttt{permutation} of $R$ by $R.\bmrank{E}$\;
        $(s_1,\dots,s_n)\gets$ \prefix{sum} of $R.D$\;
        $R[1].\bmrank{E}\gets 0$\;
        \ParFor{$i \gets 2$ \KwTo $n$}{
            $R[i].\bmrank{E}\gets s_{i-1}$\;
        }
    }
    $R \gets $ \texttt{permutation} of $R$ by $R.O$\;
    $T(\bm{F};\rank{\mathcal{E}}, O)\gets$ \texttt{expansion} of $R$ on $R.D$\;
    $T.P\gets$ \prefix{sum} of $(1,1,\dots,1)$ segmented by $T.O$\;
    \For{$\bm{E} \in \mathcal{E}$}{
        \ParFor{$i \gets 1$ \KwTo $m$}{
            \eIf{$T[i]=\bot$}{
                $T[i].\bmrank{E}\gets i$\;
            }{
                $T[i].\bmrank{E}\gets T[i].\bmrank{E}+T[i].P$\;
            }
        }
    }
    Remove attributes $O, P$ from $T$\;
    \KwRet{$T$}
\end{algorithm}

\begin{figure*}
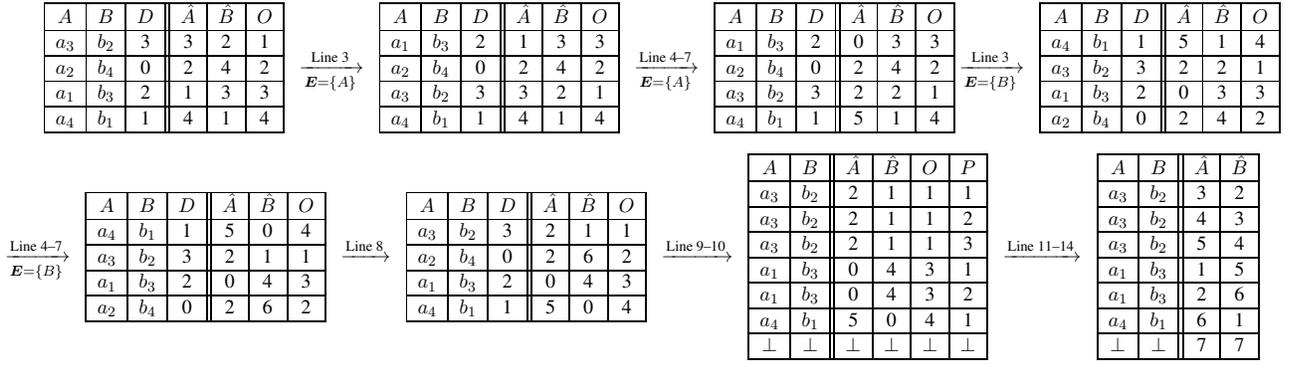

\resizebox{\textwidth}{!}{
\begin{tabular}{c}
$\phantom{\xrightarrow[\bm{E}=\{B\}]{\text{Line 4--7}}}$
\begin{tabular}{|c|c|c||c|c|c|}
\hline
${A}$ & ${B}$ & ${D}$ & ${\rank{A}}$ & ${\rank{B}}$ & $O$ \\ \hline
$a_3$ & $b_2$ & 3 & 3 & 2 & 1 \\ \hline
$a_2$ & $b_4$ & 0 & 2 & 4 & 2 \\ \hline
$a_1$ & $b_3$ & 2 & 1 & 3 & 3 \\ \hline
$a_4$ & $b_1$ & 1 & 4 & 1 & 4 \\ \hline
\end{tabular}
$~\xrightarrow[\bm{E}=\{A\}]{\text{Line 3}}~$
\begin{tabular}{|c|c|c||c|c|c|}
\hline
${A}$ & ${B}$ & ${D}$ & ${\rank{A}}$ & ${\rank{B}}$ & $O$ \\ \hline
$a_1$ & $b_3$ & 2 & 1 & 3 & 3 \\ \hline
$a_2$ & $b_4$ & 0 & 2 & 4 & 2 \\ \hline
$a_3$ & $b_2$ & 3 & 3 & 2 & 1 \\ \hline
$a_4$ & $b_1$ & 1 & 4 & 1 & 4 \\ \hline
\end{tabular}
$~\xrightarrow[\bm{E}=\{A\}]{\text{Line 4--7}}~$
\begin{tabular}{|c|c|c||c|c|c|}
\hline
${A}$ & ${B}$ & ${D}$ & ${\rank{A}}$ & ${\rank{B}}$ & $O$ \\ \hline
$a_1$ & $b_3$ & 2 & 0 & 3 & 3 \\ \hline
$a_2$ & $b_4$ & 0 & 2 & 4 & 2 \\ \hline
$a_3$ & $b_2$ & 3 & 2 & 2 & 1 \\ \hline
$a_4$ & $b_1$ & 1 & 5 & 1 & 4 \\ \hline
\end{tabular}
$\xrightarrow[\bm{E}=\{B\}]{\text{Line 3}}~$
\begin{tabular}{|c|c|c||c|c|c|}
\hline
${A}$ & ${B}$ & ${D}$ & ${\rank{A}}$ & ${\rank{B}}$ & $O$ \\ \hline
$a_4$ & $b_1$ & 1 & 5 & 1 & 4 \\ \hline
$a_3$ & $b_2$ & 3 & 2 & 2 & 1 \\ \hline
$a_1$ & $b_3$ & 2 & 0 & 3 & 3 \\ \hline
$a_2$ & $b_4$ & 0 & 2 & 4 & 2 \\ \hline
\end{tabular}\bigskip\\
$~\xrightarrow[\bm{E}=\{B\}]{\text{Line 4--7}}~$
\begin{tabular}{|c|c|c||c|c|c|}
\hline
${A}$ & ${B}$ & ${D}$ & ${\rank{A}}$ & ${\rank{B}}$ & $O$ \\ \hline
$a_4$ & $b_1$ & 1 & 5 & 0 & 4 \\ \hline
$a_3$ & $b_2$ & 3 & 2 & 1 & 1 \\ \hline
$a_1$ & $b_3$ & 2 & 0 & 4 & 3 \\ \hline
$a_2$ & $b_4$ & 0 & 2 & 6 & 2 \\ \hline
\end{tabular}
$~\xrightarrow{\text{Line 8}}~$
\begin{tabular}{|c|c|c||c|c|c|}
\hline
${A}$ & ${B}$ & ${D}$ & ${\rank{A}}$ & ${\rank{B}}$ & $O$ \\ \hline
$a_3$ & $b_2$ & 3 & 2 & 1 & 1 \\ \hline
$a_2$ & $b_4$ & 0 & 2 & 6 & 2 \\ \hline
$a_1$ & $b_3$ & 2 & 0 & 4 & 3 \\ \hline
$a_4$ & $b_1$ & 1 & 5 & 0 & 4 \\ \hline
\end{tabular}
$~\xrightarrow{\text{Line 9--10}}~$
\begin{tabular}{|c|c||c|c|c|c|}
\hline
${A}$ & ${B}$ & ${\rank{A}}$ & ${\rank{B}}$ & $O$ & $P$\\ \hline
$a_3$ & $b_2$ & 2 & 1 & 1 & 1  \\ \hline
$a_3$ & $b_2$ & 2 & 1 & 1 & 2  \\ \hline
$a_3$ & $b_2$ & 2 & 1 & 1 & 3  \\ \hline
$a_1$ & $b_3$ & 0 & 4 & 3 & 1  \\ \hline
$a_1$ & $b_3$ & 0 & 4 & 3 & 2  \\ \hline
$a_4$ & $b_1$ & 5 & 0 & 4 & 1  \\ \hline
$\perp$ & $\perp$ & $\perp$ & $\perp$ & $\perp$ & $\perp$ \\ \hline
\end{tabular}
$~\xrightarrow{\text{Line 11--14}}~$
\begin{tabular}{|c|c||c|c|c|}
\hline
${A}$ & ${B}$ & ${\rank{A}}$ & ${\rank{B}}$ \\ \hline
$a_3$ & $b_2$ & 3 & 2 \\ \hline
$a_3$ & $b_2$ & 4 & 3 \\ \hline
$a_3$ & $b_2$ & 5 & 4 \\ \hline
$a_1$ & $b_3$ & 1 & 5 \\ \hline
$a_1$ & $b_3$ & 2 & 6 \\ \hline
$a_4$ & $b_1$ & 6 & 1 \\ \hline
$\perp$ & $\perp$ & 7 & 7 \\ \hline
\end{tabular}
\end{tabular}
}
\caption{An example of Algorithm~\ref{alg:expansion_rank} with $n=4$, $m=7$, $\bm{F}=\{A,B\}$, and $\mathcal{E}=\{(A),(B)\}$}
\label{fig:expansion_rank}
\end{figure*}

\section{Free-connex Queries}
\label{sec:freeconnex}

Free-connex queries are a large class of queries made up of selection, join, projection, and group-by aggregation in a particular manner.  It is also the largest class known to be solvable in $O(n+m)$ time in plaintext, while non-free-connex queries require $O(n^w+m)$ time in the worst case, where $w>1$ is the \textit{width} of the query \cite{10.1145/2535926}.  Although there are many effective data-dependent heuristics (e.g., cost-based optimization using statistics of the data) that allow us to evaluate some non-free-connex queries faster than the worst case, they do not work for MPC, due to the security definition that the parties should learn nothing about the data. Hence, free-connex queries are indeed the best we can support under MPC on any input data, worst-case or not, unless one is willing to weaken the security definition.

SSJ \cite{secretsharedjoins} and Scape \cite{scape} only propose protocol for a two-way join operator. A natural way to compose the join operators would reveal not only the input size and output size, but also the intermediate join sizes. For example, consider a line-3 join $R_1(A, B) \Join R_2(B, C) \Join R_3(C, D)$, the naive query plan that first computes $R_{12} \gets R_1 \Join R_2$ and then computes $R_{123} \gets R_{12} \Join R_3$ would also reveals the intermediate join size $|R_{12}|$, which violates the security definition in Section~\ref{sec:model}. In this section, we show how LINQ can support free-connex queries while only revealing the input size and output size of the query. One key point is that LINQ supports $m$, the join size bound, not necessarily equal to the two-way join size, which gives the ability for intermediate results to be padded to a public size. In contrast, the two-way join protocols of SSJ and Scape neither support  this padding nor accept input relations that contain dummy tuples, so our generalization for LINQ to support free-connex queries in this section does not apply to their protocols. 

\subsection{Query Definition}
A free-connex query has the following general form:
\[ \mathcal{Q}:= \pi_{\bm{O}}^\oplus \big(\sigma_{\gamma_1} (R_1(\bm{F}_1)) \Join \cdots \Join \sigma_{\gamma_k}(R_k(\bm{F}_k)) \big), \]
Since we can easily process all the selections with $O(n)$ cost, we will ignore them in the query specification.  Also, as mentioned, projection is trivial while distinct projection is a special case of group-by aggregation, so we just focus on join and group-by aggregation below.  The group-by attributes $\bm{O}$ are also called the output attributes.

\paragraph{Full join} A full join $\mathcal{J}$ computes the joins over all input relations: $\mathcal{J}=\,\Join_{i=1}^k R_i(\bm{F_i})$. Note that the join operator $\Join$ is both commutative and associative, so the join order can be arbitrary. Let $\bm{F}=\cup_{i=1}^k \bm{F_i}$ be the set of all attributes of the relations.

\paragraph{Join tree} A join tree $\mathcal{T}$ of a full join $\mathcal{J}$ is a tree where the set of nodes are the relations $R_1,\dots,R_k$, and for any attribute $A\in \bm{F}$, the set of nodes/relations that contain $A$ are connected.  Figure \ref{fig:free_connex} shows two valid join trees. Note that not every full join has such a join tree, e.g., the triangle join $R_1(A,B)\Join R_2(B,C) \Join R_3(A,C)$. The full join that has a join tree is called an \textit{acyclic join}.  It is known that the width of any acyclic join is $1$, while the width of any cyclic join is greater than $1$, e.g., the triangle join has width $w=1.5$.

\paragraph{Annotation} We follow the same terminology from \cite{ajar, secyan}. Let $(\mathcal{S}, \oplus, \otimes)$ be a communicative semiring. Each tuple $t$ is associated with an annotation $v \in \mathcal{S}$.  If an annotation is not needed for a tuple, we set it to the $\otimes$-identity of the semiring. We conceptually add the annotations as a special attribute $V$, so the parties will store $R=R(\bm{F};V;\rank{\mathcal{E}})$ in secret-shared form.  In the implementation, this is not needed if $V$ is one of the existing attributes of $R$, or can be computed on-the-fly when $V$ is a function of the attributes, e.g., \texttt{price * quantity}. 

Next we specify how the annotations propagate through join and aggregation. 
The join $R\Join^\otimes S$ of two annotated relations $R,S$ also returns an annotated relation, where the tuple joined by $t_R\in R$ and $t_S\in S$ has the annotation $t_R.V\otimes t_S.V$.
The group-by-aggregation operator, $\pi_{\bm{O}}^\oplus$, over an annotated relation $R$, computes the $\oplus$-aggregate of each group of tuples grouped by $\bm{O}$.

\paragraph{Free-connex query} A query $\mathcal{Q}=\pi_{\bm{O}}^\oplus (\mathcal{J})$ is \textit{free-connex} if (1) $\mathcal{J}$ is acyclic, and (2) there exists a join tree $\mathcal{T}$ of $\mathcal{J}$ such that for any $A\in\bm{O}$ and $B\in\bm{F}-\bm{O}$, $\TOP(B)$ is not an ancestor of $\TOP(A)$ in $\mathcal{T}$, where $\TOP(X)$ denotes the highest node in $\mathcal{T}$ containing attribute $X$. Such a join tree is called a free-connex join tree.
In the two join trees in Figure~\ref{fig:free_connex}, the one on the left is free-connex. 
The one on the right is not because $\TOP(F)=R_3$ is the ancestor of $\TOP(A)=R_1$.  

Two special cases are of interest: When $\bm{O} = \bm{F}$, the query becomes a full join without aggregation; when $\bm{O}=\emptyset$, all join results are in one group, i.e., the query is a complete aggregation without group-by.  


\begin{example}
\label{ex:tpch3}
    TPC-H Query 3 is a typical free-connex query:
    \begin{center}
    \begin{tabular}{c}
    \begin{lstlisting}
SELECT o_orderkey, o_orderdate, o_shippriority,
       SUM(l_extendedprice * (1 - l_discount))
  FROM customer, orders, lineitem
 WHERE c_custkey = o_custkey AND l_orderkey = o_orderkey
   AND c_mktsegment = 'BUILDING'
   AND o_orderdate < date '1995-03-13'
   AND l_shipdate > date '1995-03-15'
GROUP BY o_orderkey, o_orderdate, o_shippriority;
    \end{lstlisting}
    \end{tabular}
    \end{center}
After attribute renaming, this query can be written using relational algebra as (we omit the selection operators):
\begin{align*}
&\pi^\oplus_{\mathtt{orderkey,orderdate,shippriority}} \big( \mathtt{customer(custkey,mktsegment)} \\
 &\Join \mathtt{order(orderkey,custkey,orderdate,shippriority)}  \\
& \Join \mathtt{lineitem(orderkey,shipdate)}\big) 
\end{align*} 
    We use the semiring $(\mathbb{R}^{+}, +, \times)$, and the annotations associated with \texttt{lineitem} are \texttt{l\_extendedprice*} \texttt{(1-l\_discount)}, whereas the annotations in all other relations are 1.  This query is free-connex, using the join tree that puts \texttt{orders} as the root with \texttt{lineitem} and \texttt{customer} as its two children.  Note that all three output attributes are in the root relation. \qed   
\end{example}

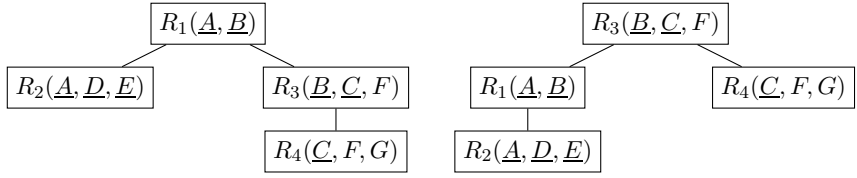
\begin{figure}
    \centering
    \begin{tikzpicture}[sibling distance=5em, level distance=3em, scale=0.85,
      every node/.style = {shape=rectangle, draw, align=center, scale=0.85}]]
    \node (root) at(0, 0) {$R_1(\underline{A}, \underline{B})$};
    \node (left) at(-2, -1) {$R_2(\underline{A}, \underline{D}, \underline{E})$};
    \node (right) at(2, -1) {$R_3(\underline{B}, \underline{C}, F)$};
    \node (rightleaf) at(2, -2) {$R_4(\underline{C}, F, G)$};

    \node (broot) at(7, 0) {$R_3(\underline{B}, \underline{C}, F)$};
    \node (bleft) at(5, -1) {$R_1(\underline{A}, \underline{B})$};
    \node (bright) at(9, -1) {$R_4(\underline{C}, F, G)$};
    \node (bleftleaf) at(5, -2) {$R_2(\underline{A}, \underline{D}, \underline{E})$};
    
    \draw[-] (root) -- (left);
    \draw[-] (root) -- (right);
    \draw[-] (right) -- (rightleaf);

    \draw[-] (broot) -- (bleft);
    \draw[-] (broot) -- (bright);
    \draw[-] (bleft) -- (bleftleaf);
\end{tikzpicture}
    \caption{Two valid join trees for the join-aggregate query $\mathcal{Q}=\pi_{\{A, B, C, D\}}^\oplus\left(\Join_{i\in[4]}^\otimes R_i\right)$ with the output attributes underlined. The left one is a free-connex join tree for $\mathcal{Q}$.}
    \label{fig:free_connex}
\end{figure}

\subsection{The Protocol}
\label{sec:yan}
Assume for simplicity that the size of each relation $|R_i|=n$ and the output size $|\mathcal{Q}|=m$ is given.
The Yannakakis algorithm \cite{yannakakis} and its aggregation version \cite{ajar,faqpaper,sizebounds} can evaluate any free-connex query in $O(n + m)$ time in plaintext. 
It consists of the following steps:

\begin{enumerate}
    \item \textit{Reduce.} The \textit{reduce} step works by visiting the join tree $\mathcal{T}$ in a bottom-up order. For each node $R(\bm{F_R}) \in \mathcal{T}$ and its parent  $P(\bm{F_P})$ and in $\mathcal{T}$, let $\bm{F} = (\bm{O} \cup \bm{F_P}) \cap \bm{F_R}$. We update $R \gets \pi_{\bm{F}}^{\oplus} (R)$. If $\bm{F} \subseteq \bm{F_P}$, we update $P \gets P \Join^{\otimes} R$, and then remove $R$ from $\mathcal{T}$; otherwise, stop going upward.
    This pass terminates when either the root is the only node in $\mathcal{T}$, or all attributes in $\mathcal{T}$ are output attributes. In the former case, the algorithm terminates, and the relation at the root node is exactly the result of $\mathcal{Q}$. In the latter case, since all attributes are in $\bm{O}$, the query degenerates to a full join and the remaining $\mathcal{T}$ is its join tree. It then goes to the next step. 
    \item \textit{Semi-joins.} Next, the algorithm uses two passes of semi-joins on the join tree to remove all the \textit{dangling tuples}, i.e., tuples that do not contribute to the full join result. Specifically, we first use a bottom-up pass to check each pair of relation $R$ and its parent $P$ and update $P \gets P \ltimes R$; then use a top-down pass to check each pair of relation $R$ and its parent $P$ and update $R \gets R \ltimes P$.  
    \item \textit{Joins.} Finally, compute the full join in a bottom-up pass: While there is more than one node in $\mathcal{T}$, pick a leaf relation $R$ and its parent $P$, update $P \gets P \Join^{\otimes} R$ with join size upper bound $m$, and remove $R$ from tree.   The process terminates the only node in $\mathcal{T}$ is the root, which is the query result. 
\end{enumerate}

Given a free-connex join tree, the Yannakakis algorithm induces a query plan that consists of join, group-by-aggregation, and semi-join operators.  To generalize our protocol for join to support a free-connex join, the key is to make sure that the rank attribute $\bmrank{E}$ that any operator relies on are available. As a group-by-aggregation or a semi-join does not introduce any new attribute to the relation, so it automatically satisfies the condition. In the \textit{reduce} step, the join operator is always in the form $P \Join^{\otimes} \pi_{\bm{F}}^{\oplus} (R)$ with $\bm{F}\subseteq \bm{F_P}$, so it is actually a semi-join and thus does not introduce any new attribute to $P$.

Then we move to step (3) where all relations only contain output attributes. Consider any join $P \Join^{\otimes} R$. We denote $P'(\bm{F_P}')$ and $R'(\bm{F_R}')$ as the two relations at the initial state of step (3). Obviously $\bm{F_P'} \subseteq \bm{F_P}$ and $\bm{F_R'} \subseteq \bm{F_R}$. We prove that both  $\bm{F_P}\cap \bm{F_R}\subseteq \bm{F_P}'$ and $\bm{F_P}\cap \bm{F_R}\subseteq \bm{F_R}'$ hold, so the rank attributes have been prepared and the join can be computed. If the former condition does not hold, then there is an attribute $A\in \bm{F_R}$ such that $A\in\bm{F_P}-\bm{F_P}'$. Since $A\notin \bm{F_P}'$, there must have been a previous step that computes the join between $P$ and its another child except $R$. Both that child and $R$ have attribute $A$, but $P$ does not. This violates the property of a join tree that all nodes containing $A$ is connected (and one can prove that this property always holds during the execution of the algorithm). We can deduce a contradiction from the latter condition in a similar way.

Below we analyze the complexity.  Every operator in the \textit{reduce} step or the \textit{semi-joins} step takes one or two relations each with size $n$ and outputs a relation with size $n$, so the costs are all $O(n)$ and the number of rounds is $O(\log n)$. In the \textit{join} step, since dangling tuples have been removed, any intermediate join size is bounded by the final join size $m$. 
In fact, we should use $m$ as the join size of each join so as to hide the intermediate join sizes.  So the cost is $O(n+m)$ and the number of rounds is $O(\log(n+m))$.

\begin{theorem}
   LINQ can evaluate any free-connex query with complexity $O(n+m)$ under the 3PC model. 
\end{theorem}

\begin{example}\label{ex:yan}
Consider evaluating the query in Figure~\ref{fig:free_connex} using the join tree on the left. The query plan induced by this join tree is as follows:
\begin{enumerate}
\item \textit{Reduce.} $R_4(C,F)\gets \pi_{C,F}^\oplus(R_4)$; $R_3(B,C,F)\gets R_3\Join^\otimes R_4$; $R_3(B,C)\gets \pi_{B,C}^\oplus (R_3)$. Figure~\ref{fig:yan} shows how the join tree changes during this step. 
\item \textit{Semi-joins.} $R_1\gets R_1\ltimes R_2$; $R_1\gets R_1\ltimes R_3$; $R_2\gets R_2\ltimes R_1$; $R_3\gets R_3\ltimes R_1$.
\item \textit{Joins.} $R_1(A,B,D,E)\gets R_1\Join^\otimes R_2$; $R_1(A,B,C,D,E)\gets R_1\Join^\otimes R_3$; $R_1$ is the output.
\end{enumerate}
\begin{figure}
\begin{center}
\begin{tabular}{l}
\begin{tikzpicture}[sibling distance=5em, level distance=3em, scale=0.85,
      every node/.style = {shape=rectangle, draw, align=center, scale=0.85}]]
\node (root) at(0, 0) {$R_1(\underline{A}, \underline{B})$};
\node (left) at(-2, -1) {$R_2(\underline{A}, \underline{D}, \underline{E})$};
\node (right) at(2, -1) {$R_3(\underline{B}, \underline{C}, F)$};
\node (rightleaf) at(2, -2) {$R_4(\underline{C}, F, G)$};
\draw[-] (root) -- (left);
\draw[-] (root) -- (right);
\draw[-] (rightleaf) -- (right);

\coordinate  (c0) at(3.5,-1) ;
\coordinate  (c1) at(4.5,-1) ;
\draw[->][thick] (c0) -- (c1) ;

\node (1root) at(8, 0) {$R_1(\underline{A}, \underline{B})$};
\node (1left) at(6, -1) {$R_2(\underline{A}, \underline{D}, \underline{E})$};
\node (1right) at(10, -1) {$R_3(\underline{B}, \underline{C}, F)$};
\node (1rightleaf) at(10, -2) {$R_4(\underline{C}, F)$};
\draw[-] (1root) -- (1left);
\draw[-] (1root) -- (1right);
\draw[-] (1rightleaf) -- (1right);



\coordinate  (c2) at(8,-2.5) ;
\coordinate  (c3) at(8,-3.5) ;
\draw[->][thick] (c2) -- (c3) ;

\node (2root) at(8, -4.3) {$R_1(\underline{A}, \underline{B})$};
\node (2left) at(6, -5.3) {$R_2(\underline{A}, \underline{D}, \underline{E})$};
\node (2right) at(10, -5.3) {$R_3(\underline{B}, \underline{C},F)$};
\draw[-] (2root) -- (2left);
\draw[-] (2root) -- (2right);

\coordinate  (c4) at(4.5,-4.7) ;
\coordinate  (c5) at(3.5,-4.7) ;
\draw[->][thick] (c4) -- (c5) ;

\node (3root) at(0, -4.3) {$R_1(\underline{A}, \underline{B})$};
\node (3left) at(-2, -5.3) {$R_2(\underline{A}, \underline{D}, \underline{E})$};
\node (3right) at(2, -5.3) {$R_3(\underline{B}, \underline{C})$};
\draw[-] (3root) -- (3left);
\draw[-] (3root) -- (3right);

\end{tikzpicture}
\end{tabular}
\end{center}
\caption{Examples of the \textit{reduce} step of the Yannakakis algorithm}
\label{fig:yan}
\end{figure}
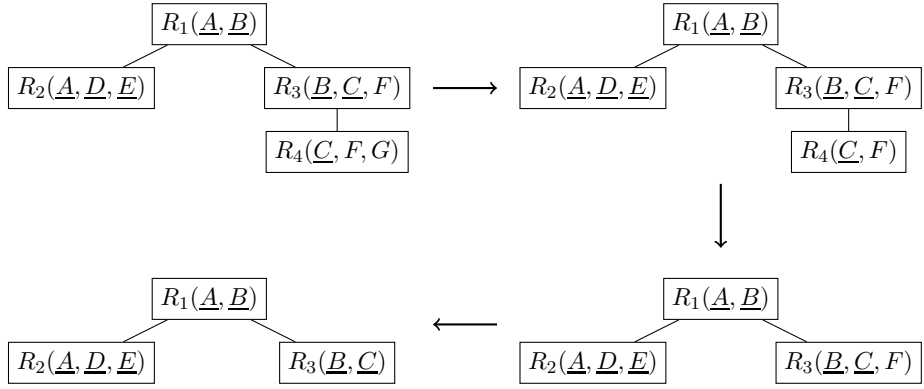
\end{example}

\subsection{Compute Output Size}\label{sec:compute_out}
We have so far assumed that our protocol is given the output size $m$. In this section, we show how to compute the value of $m$ with $O(n)$ cost and $O(\log n)$ rounds. 

Observe that the first two steps of the algorithms do not require $m$, which we still run as before.  After these two steps, we obtain a full acyclic join $\mathcal{J}'$ without dangling tuples, and its output size is the same as that of the original query $\mathcal{Q}$.  To compute the output size of $\mathcal{J}'$, we turn it into another free-connex query $\mathcal{Q}'=\pi_{\emptyset}^+(\mathcal{J}')$ where the annotations of all tuples are set to $1$, and we use the semiring $(\mathbb{Z},+,\times)$, i.e., the ring of integers.  Since $\mathcal{Q}'$ has no output attributes, on this query the algorithm will terminate in step (1), which returns $m=|\mathcal{Q}|$.  Then we can continue with step (3) on $\mathcal{Q}$.

\begin{example}
Consider the free-connex query in Example~\ref{ex:yan}. After the first 2 steps, we have $\mathcal{J}'=R_1(A,B)\Join^\otimes R_2(A,D,E)\Join^\otimes R_3(B,C)$, and its join tree is the last tree in Figure~\ref{fig:yan}. Then we run just the reduce step on $\mathcal{Q}'$:  $R_2(A)\gets \pi_{A}^+(R_2)$; $R_1(A,B)\gets R_1\Join^\times R_2$; $R_3(B)\gets \pi_{B}^+ (R_3)$; $R_1(A,B)\gets R_1\Join^\times R_3$; $m\gets\pi_{\emptyset}^+(R_1)$.
\end{example}

\section{System Implementation}
Based on LINQ, we have built a system prototype for evaluating any free-connex query under the 3PC model.  
Since our protocol works based on a free-connex join tree, the parties need to select one free-connex join tree when the query is given. This selection strategy is detailed in Section~\ref{sec:query_optimize}. Then they execute the query by the protocol based on the selected free-connex join tree, which is described in Section~\ref{sec:exec}.

\subsection{Query Plan Optimization}\label{sec:query_optimize}
Given any free-connex query $\mathcal{Q}$, our theory says that the query plan induced by any join tree of $\mathcal{Q}$ has the same $O(n+m)$ asymptotic cost.  However, the hidden constant (including the number of attributes and rank attributes of each relation, the bit-length of different attributes, etc.) can vary greatly. 
Since oblivious operators are all data-independent, the cost of them can be calculated precisely given input and output sizes. In this section, we introduce a more accurate cost model to estimate the concrete performance of different query plans, by taking the number of attributes (including the rank attributes) of each relation into consideration, which allows us to choose the best plan for execution.

For any relation $R(\bm{F};\rank{\mathcal{E}})$, $|\bm{F}|$ is the number of original attributes and $|\mathcal{E}|$ is the number of rank attributes. The cost of order-by operator is linear to $|\bm{F}|+|\mathcal{E}|$, as it permutes the whole relations. The projection operator incurs no communication. The cost of the selection operator comes from two steps. The first step is to apply the predicates, which is at most linear to $|\bm{F}|$. The second step is applying Algorithm~\ref{alg:update_rank} to update the rank attributes. We note that for each $\bm{E}\in\mathcal{E}$, updating $\bmrank{E}$ actually does not require permuting the whole relation. Instead, we first compute the \texttt{permutation} of $(O,Z)$ by $\bmrank{E}$, where $O=[n]$ and $Z$ is the dummy marker. Then we compute the \texttt{compaction} of it, set the $i$-th element of $\bmrank{E}$ to $i$ for each $i\in[n]$, and permute it to the original order by computing the \texttt{permutation} of $\bmrank{E}$ by $O$. Since these permutations and the compaction are related to only constant number of attributes, the total cost of Algorithm~\ref{alg:update_rank} is linear to $|\mathcal{E}|$. Therefore, the total cost of the selection operator is at most linear to $|\bm{F}|+|\mathcal{E}|$.
The cost of group-by-aggregation operator is also at most linear to $|\bm{F}|+|\mathcal{E}|$ due to the same reason.

The above trick also works for updating the rank attributes in expansion. Specifically, for each rank attribute $\bm{E}\in\mathcal{E}$, we update $\bmrank{E}$ in Line 3--7 in Algorithm~\ref{alg:expansion_rank}. The permutation actually does not need to be performed over the whole relation $R$, but only over $O$ and $D$, so the total cost of updating the rank attributes in expansion is linear to $|\mathcal{E}|$. The total cost of Algorithm~\ref{alg:expansion_rank} is therefore at most linear to $|\bm{F}|+|\mathcal{E}|$.

Now we assume the two input relations to a (semi-)join operator are $R(\bm{F_R};\rank{\mathcal{E}}_R)$ and $S(\bm{F_S};\rank{\mathcal{E}}_S)$. The semi-join operator has cost linear to $|\bm{F_R}|+|\mathcal{E}_R|$, where $|\bm{F_R}|$ is for the extended \texttt{intersection} on the common attributes of $R$ and $S$, and $|\mathcal{E}_R|$ is for updating rank attributes by Algorithm~\ref{alg:update_rank}. For the join operator, the cost is dominated by computing the \texttt{expansion} of $R$ and $S$, so the total cost is linear to $|\bm{F_R}|+|\mathcal{E}_R|+|\bm{F_S}|+|\mathcal{E}_S|$.

In conclusion, all operators have cost at most linear to the total number of attributes (including the rank attributes) of the input relations (despite semi-join which is linear to the first relation only). We use this value multiplied by the data complexity as the cost estimation of an operator. The cost of a join tree is therefore the total costs of the operators in the query plan it induces. 

\begin{example}
    Recall the free-connex query in Example~\ref{ex:yan}. 
    Consider another free-connex join tree, in which we put $R_2$ as the root, i.e., the join tree becomes a line $R_2-R_1-R_3-R_4$. To compare the two plans, we measure their total costs in the \textit{semi-joins} step and \textit{final join} step, because they have the same query plan in the \textit{reduce} step. The details are shown in Table~\ref{tab:plan_cost}, which suggest we should choose the new free-connex join tree.

    \begin{table}
\centering
\begin{tabular}{|c|c|}
\hline
Operation & Cost \\ \hline
$R_1(A,B;\rank{A},\rank{B})\gets R_1(A,B;\rank{A},\rank{B})\ltimes R_2(A,D,E;\rank{A})$ & $4n$ \\ \hline
$R_1(A,B;\rank{A},\rank{B})\gets R_1(A,B;\rank{A},\rank{B})\ltimes R_3(B,C;\rank{B})$ & $4n$ \\ \hline
$R_2(A,D,E;\rank{A})\gets R_2(A,D,E;\rank{A}) \ltimes R_1(A,B;\rank{A},\rank{B})$ & $4n$ \\ \hline
$R_3(B,C;\rank{B})\gets R_3(B,C;\rank{B}) \ltimes R_1(A,B;\rank{A},\rank{B})$ & $3n$ \\ \hline
$R_1(A,B,D,E;\rank{B})\gets R_1(A,B;\rank{A},\rank{B})\Join^\otimes R_2(A,D,E;\rank{A})$ & $8m$ \\ \hline
$R_1(A,B,C,D,E)\gets R_1(A,B,D,E;\rank{B})\Join^\otimes R_3(B,C;\rank{B})$ & $8m$ \\ \hline
Total cost of old join tree & $15n+16m$ \\\hline
\end{tabular}
\bigskip\\
\begin{tabular}{|c|c|}
\hline
Operation & Cost \\ \hline
$R_1(A,B;\rank{A},\rank{B})\gets R_1(A,B;\rank{A},\rank{B})\ltimes R_3(B,C;\rank{B})$ & $4n$ \\ \hline
$R_2(A,D,E;\rank{A})\gets R_2(A,D,E;\rank{A}) \ltimes R_1(A,B;\rank{A},\rank{B})$ & $4n$ \\ \hline
$R_1(A,B;\rank{A},\rank{B})\gets R_1(A,B;\rank{A},\rank{B})\ltimes R_2(A,D,E;\rank{A})$ & $4n$ \\ \hline
$R_3(B,C;\rank{B})\gets R_3(B,C;\rank{B}) \ltimes R_1(A,B;\rank{A},\rank{B})$ & $3n$ \\ \hline
$R_1(A,B,C;\rank{A})\gets R_1(A,B;\rank{A},\rank{B})\Join^\otimes R_3(B,C;\rank{B})$ & $7m$ \\ \hline
$R_2(A,B,C,D,E)\gets R_2(A,D,E;\rank{A})\Join^\otimes R_1(A,B,C;\rank{A})$ & $8m$ \\ \hline
Total cost of new join tree & $15n+15m$ \\\hline
\end{tabular}
\caption{The costs of two join trees}
\label{tab:plan_cost}
\end{table}
\end{example}

So far we have assumed that all input relations have the same size $n$. Asymptotically speaking, there is no loss of generality, since we can simply add dummy tuples to smaller relations, which only increases the cost by a constant factor.  We do not really need to do this.  Our join algorithm can actually be easily modified to support joining two relations of different sizes.

\subsection{Query Execution}\label{sec:exec}
Upon receiving a  query, 
the parties enumerate all join trees and find the optimal query plan as described above.  
Our execution engine is built on top of ABY3 \cite{aby3}\footnote{\url{https://github.com/ladnir/aby3}}, which is secure against semi-honest adversaries and already provides the implementation of some three-party primitives like circuit-based computations and intersection.  For those that are not, such as prefix sum circuit, permutation, compaction and expansion, we have done our own implementation in C++.  The default bit length of each attribute is 64. Security parameters are set to $\kappa = 128$ and $\sigma = 40$.
After the query evaluation is done, the parties transmit the secret shares of the query result to the client, who can then reconstruct the final query result, or produce further computation tasks.

\section{Experiments}\label{sec:exp}

\subsection{Experiment Setup}
All experiments are measured on three servers, each equipped with a 2.3GHz Core i9 CPU and 32GB memory.  They are connected under a LAN, with 0.1ms network delay and 1Gb/s bandwidth.
We report the longest running time and the maximum communication costs (including data sent and received) among the three servers.
All results are the average over 10 runs.
The code of our system, as well as OptScape, is available at \url{https://anonymous.4open.science/r/LINQ}.

\paragraph{Baselines}
We compare LINQ with plaintext algorithm, OptScape, and SECYAN \cite{secyan}. The running time of plaintext algorithm is measured by PostgreSQL\footnote{\url{https://www.postgre.org/}}, where the communication cost refers to the total input and output size. 
OptScape is an optimized version of Scape, where we replace their $O(n \log^2 n)$ bitonic sorting network with the $O(n \log n)$ MPC sorting protocol \cite{hamada2013sort}. Note that our experimental results on $n = 2^{16}$ tuples indicate that OptScape achieves a speedup of 10x over Scape and 5x over SSJ, thus it serves as the state of the art baseline under 3PC model. The primitives of OptScape are the version that defend against semi-honest adversary, same as LINQ. Besides, since OptScape's protocol is for three-server without plaintext preprocessing on the input relations, for fairness comparison, we also do not count its cost for sorting the input relations on the join key. We also compare with SECYAN \cite{secyan} because it is state-of-the-art protocol that also supports free-connex queries. However, its security model is incomparable to LINQ's since it is under 2PC (stronger) but with an extra requirement that one party is the query receiver (weaker). 

\subsection{TPC-H Queries}
The first set of queries\footnote{The queries we tested are  in  \url{https://anonymous.4open.science/r/LINQ/queries.md}.} we tested in the experiments are taken from the TPC-H benchmark\footnote{\url{https://www.tpc.org/tpch/}}.  All joins in the TPC-H queries are PK-FK joins, so the two-way join protocol of OptScape can be directly composed without breaching security guarantee. 
The TPC-H datasets generated by its database population program\footnote{\url{https://github.com/electrum/tpch-dbgen}} with scales varing from 0.001 to 1. 
\begin{itemize}
    \item \textbf{Q3} / \textbf{Q10} / \textbf{Q18}: The three queries are taken from SECYAN \cite{secyan}, and \textbf{Q3} has been introduced in Example~\ref{ex:tpch3}.
    \item \textbf{Q11}: A complex group-by aggregation query of three relations which needs two semijoins to transmit the selection condition from \texttt{nation} to \texttt{partsupp}.
    \item \textbf{Q3F}: A full join of the involved 3 relations in TPC-H Q3.
    \item \textbf{Q5F}: A full join of the involved 6 relations in TPC-H Q5.
\end{itemize}

Each relation is owned by one of the three parties, and we tested the worst possible way to partition the relations
such that the two relations of each join are owned by different parties so that there is no local join optimization. 

The experimental results are presented in Table~\ref{tab:comp_tpch}.
We see that the running time of our system is consistently 2x--4x lower than that of OptScape, and the communication cost is 4x--8x lower.  More importantly, our protocol demonstrates linear growth rates, as predicted by the theory.  In contrast, OptScape exhibits a logarithmic factor growth, so the gap between the two will grow wider as the data size further increases.  
SECYAN also demonstrates a logarithmic factor growth in communication cost: 2x-3x larger than ours. Its time cost is similar to ours when the data size is small, but costs a lot when the data size increases due to its extra logarithm factor.
With our linear-complexity protocol, the running time of MPC query processing is now around 100x slower than plaintext, with a communication cost that is <100x that of the input and output size.  This has made MPC processing more practical than before.

\begin{table*}
\renewcommand\arraystretch{1.2}
    \resizebox{\textwidth}{!}{
	\centering
	\begin{tabular}{|c||c|c|c|c||c|c|c|c||c|c|c|c|}
            \hline
            \multicolumn{13}{|c|}{\textbf{Time (s)}} \\
		\hline
		Query & \multicolumn{4}{|c||}{\textbf{Q3}} & \multicolumn{4}{|c||}{\textbf{Q10}} & \multicolumn{4}{|c|}{\textbf{Q18}} \\
		\hline
		Scale & 0.001 & 0.01 & 0.1 & 1 & 0.001 & 0.01 & 0.1 & 1 &  0.001 & 0.01 & 0.1 & 1 \\
		\hline
		Plaintext & 0.003 & 0.02 & 0.30 & 1.78 & 0.01 &	0.18 & 1.08 & 4.02 & 0.01 & 0.09 & 0.99 & 3.05 \\
		OptScape & 1.14 & 4.33 & 31.25 & 383.50  & 1.98 & 4.39 & 32.52 & 373.36 & 2.01 & 4.43 & 31.89 & 370.90 \\
        SECYAN & 0.34 & 2.23 & 20.36 & 189.69 & 0.34 & 2.37 & 21.33 & 192.32 & 0.58 & 3.85 & 34.21 & 287.07 \\
		LINQ & 0.35 & 0.94 & 8.36 & 107.79  & 0.41 & 0.95 & 8.23 & 103.57 & 0.38 & 0.95 & 8.83 & 112.00 \\
            \hline \hline
            Query & \multicolumn{4}{|c||}{\textbf{Q11}} & \multicolumn{4}{|c||}{\textbf{Q3F}} & \multicolumn{4}{|c|}{\textbf{Q5F}} \\
		\hline
		Scale & 0.001 & 0.01 & 0.1 & 1 & 0.001 & 0.01 & 0.1 & 1 & 0.001 & 0.01 & 0.1 & 1 \\
		\hline
		Plaintext & 0.0005 & 0.004 & 0.02 & 0.28  & 0.02 & 0.22 & 1.92 & 19.2  &  0.019 & 0.21 & 1.74 & 18.19 \\
		OptScape & 1.34 & 2.11 & 6.10 & 47.94  & 2.66 & 6.49 & 49.00 & 645.40 & 5.14 & 12.20 & 86.97 & 1155.63 \\
		LINQ & 0.44 & 0.62 & 2.14 & 21.16  & 1.00 & 3.04 & 25.37 & 317.87 & 2.03 & 5.06 & 46.01 & 541.16 \\
            \hline 
            \multicolumn{13}{c}{} \\
            \hline
            \multicolumn{13}{|c|}{\textbf{Communication (MB)}} \\
		\hline
		Query & \multicolumn{4}{|c||}{\textbf{Q3}} & \multicolumn{4}{|c||}{\textbf{Q10}} & \multicolumn{4}{|c|}{\textbf{Q18}} \\
		\hline
		Scale & 0.001 & 0.01 & 0.1 & 1 & 0.001 & 0.01 & 0.1 & 1 & 0.001 & 0.01 & 0.1 & 1 \\
		\hline 
		Plaintext & 0.23 & 2.32 & 23.13 & 231.20  & 0.13 & 1.32 & 13.17 & 131.63 & 0.14 & 1.39 & 13.86 & 138.49 \\
		OptScape & 59.55 & 752.11 & 9572.57 & 113636.14  &  59.07 & 756.63 & 9584.31 & 113562.32 & 60.93 & 753.15 & 9487.35 & 114986.29 \\
        SECYAN & 15.95 & 185.50 & 2118.29 & 24004.68 & 12.33 & 142.82 & 1638.70 & 18607.46  & 23.51 & 266.82 & 2999.72 & 33510.06 \\
        LINQ & 7.70 & 76.82 & 767.77 & 7682.20  & 6.67 & 66.55 & 665.07 & 6655.33 & 7.36 & 73.44 & 734.02 & 7344.96 \\
            \hline \hline
            
		Query & \multicolumn{4}{|c||}{\textbf{Q11}} & \multicolumn{4}{|c||}{\textbf{Q3F}} & \multicolumn{4}{|c|}{\textbf{Q5F}} \\
		\hline
		Scale & 0.001 & 0.01 & 0.1 & 1 & 0.001 & 0.01 & 0.1 & 1 & 0.001 & 0.01 & 0.1 & 1 \\
		\hline 
		Plaintext & 0.02 & 0.19 & 1.89 & 19.22 & 0.41 & 4.15 & 41.46 & 414.35 & 0.44 & 4.38 & 43.76 & 437.64 \\
		OptScape & 6.97 & 90.12 & 1192.78 & 12966.69  & 78.27 & 921.01 & 10886.30 & 129971.91 &  129.83 & 1622.45 & 19006.96 & 226924.04 \\
        LINQ & 1.99 & 19.26 & 192.17 & 1921.44  &  20.88 & 208.21 & 2078.34 & 20774.26 & 34.42 & 341.93 & 3411.10 & 34092.46 \\
        \hline
	\end{tabular}
    }
	\caption{Costs of different protocols for TPC-H queries}
	\label{tab:comp_tpch}
\end{table*}

\subsection{Graph Queries}
Next, we tested some queries with non-key joins.  We used the bitcoin-alpha dataset (called `bitcoin') \cite{kumar2016edge, kumar2018rev2}, which is a directed graph storing who-trusts-whom network of people who trade using Bitcoin on a platform.
There are $n = 	24186$ edges in the graph; each tuple represents an edge, containing four attributes: source, target, rating and time. 
We test the following query with tuned parameter $k$ for selection conditions.
\begin{center}
\begin{tabular}{c}
\begin{lstlisting}[mathescape=true]
 SELECT b1.source, b1.target, b2.target, b3.target
   FROM bitcoin AS b1, bitcoin AS b2, bitcoin AS b3
  WHERE b1.target = b2.source AND b2.target = b3.source
    AND b1.rating >= $k$ 
    AND b2.rating >= $k$ 
    AND b3.rating >= $k$;
\end{lstlisting}
\end{tabular}
\end{center}

It is a self-join query on three copies of the relation, and we assume each party holds one of the copy. We tried this query with parameters $k=6,5,4,3$ which lead to different output sizes $m$.
The results are in Table~\ref{tab:comp_graph}. We do not compare with OptScape because, as mentioned in Section~\ref{sec:freeconnex}, it reveals the intermediate join size which breaches the security guarantee. We also do not compare with SECYAN because it does the join locally by one of the party. We tried to compare with Secrecy \cite{secrecy} and QCircuit \cite{wang2022query}, but neither finished the query in an hour. Therefore, we only report the results of our protocol and PostgreSQL. We also counted the cost of computing the output size of these queries in LINQ. The running time is around 1s and the communication cost is 90MB, despite the selection conditions. This is because the query that computes the output size is a free-connex query with fixed input size $n=24186$ and output size $m=1$ despite the selection conditions.

\begin{table}
	\centering
	\begin{tabular}{|c||c|c|c|c|}
            \hline
            & \multicolumn{4}{|c|}{\textbf{Parameter}} \\\hline
		$k$ & 6 & 5 & 4 & 3 \\
            $m$ & 21151 & 94920 & 234827 & 887494 \\\hline\hline
            & \multicolumn{4}{|c|}{\textbf{Time (s)}} \\\hline
		Plaintext & 0.10 & 0.27 & 0.45 & 1.15 \\
		LINQ & 5.04 & 11.84 & 25.75 & 93.64  \\\hline\hline
            & \multicolumn{4}{|c|}{\textbf{Comm (MB)}} \\\hline
		Plaintext & 2.30 & 4.56 & 8.83 & 28.74 \\
		LINQ & 356.73 & 952.30 & 2085.10 & 7369.73  \\\hline
	\end{tabular}
	\caption{Costs of graph query varying parameter $k$}
	\label{tab:comp_graph}
\end{table}

\section{Extend to the Three-Server Model}
\label{sec:3server}
The three-server model is a variant of 3PC, where the input is given in secret-shared form.  In practice, three different cloud service providers can play the role of the three servers, and each data owner just uploads its data to them in secret-shared form for query processing.  This variant has recently got some traction, since it can support a large number of data owners. 

LINQ, strictly speaking, does not work in the three-server model, since we cannot do consistent-sorting in $O(n)$ time on secret-shared data.  Nevertheless, there are two simple workarounds.  First, each owner can do a consistent-sort on its data to obtain the ranks, and then upload them to the three servers in secret-shared form.  Note that the consistent-sort does not increase the asymptotic cost for the data owner, which is still $O(n)$.  Note that, after this step, all remaining steps of LINQ work in the three-server model.  If the data owners are not willing to pay this extra cost, we can still replace the consistent-sort with secure sorting \cite{efficient3pcsorting} over secret-shared data, which incurs $O(n\log n)$ cost.  Note that this secure sorting only needs to be done once to obtain the ranks, which can be used to support any number of queries, each with $O(n + m)$ cost.  Thus, after $\Omega(\log n)$ queries, the amortized cost per query  becomes linear. 

\section{Conclusion}
In this paper, we have presented LINQ, the first linear-complexity join protocol under the three-party model of MPC.  We have then extended it to support any free-connex query and built a system prototype.  One intriguing question is whether linear complexity is also achievable with only two parties, which would further lower the deployment effort of such MPC systems.  In fact, most of the primitives LINQ uses have corresponding two-party protocols with same complexity, except \texttt{permutation}, for which the best practical two-party protocols have $O(n \log n)$ complexity \cite{secretsharedshuffle}, and whether this can be brought down to $O(n)$ is still an open problem in the security literature. Any improvement on these primitives will also bring an improvement to design LINQ under the two-party model. 

\bibliographystyle{alpha}

\newcommand{\etalchar}[1]{$^{#1}$}

\end{document}